%% file: main.tex
\documentclass[11pt]{article}

\usepackage[margin=1in]{geometry}
\usepackage{amsmath,amsthm,amssymb,mathtools}
\usepackage{cite}
\usepackage{array}
\usepackage[font=small,labelfont=bf]{caption}

\usepackage[ruled,norelsize,vlined]{algorithm2e}
\DontPrintSemicolon

\usepackage{accents}

\usepackage{etoolbox} 

\renewcommand{\gets}{:=}

\newcommand{\Otilde}[1]{\tilde{O}(#1)}

\newcommand{\alglabel}[1]  {\label{alg:#1}}
\newcommand{\algref}[1]    {Algorithm~\ref{alg:#1}}

\newcommand{\seclabel}[1]  {\label{sec:#1}}
\newcommand{\secref}[1]    {Section~\ref{sec:#1}}

\newcommand{\secrefs}[2]   {Sections~\ref{sec:#1}--\ref{sec:#2}}
\newcommand{\lemlabel}[1]  {\label{lem:#1}}
\newcommand{\lemref}[1]    {Lemma~\ref{lem:#1}}

\newcommand{\thmlabel}[1]  {\label{thm:#1}}
\newcommand{\thmref}[1]    {Theorem~\ref{thm:#1}}
\newcommand{\corlabel}[1]  {\label{cor:#1}}
\newcommand{\corref}[1]    {Corollary~\ref{cor:#1}}

\newcommand{\defn}[1]      {\boldmath\textbf{\emph{#1}}\unboldmath}
\newcommand{\id}[1]        {\ifmmode\mathit{#1}\else\textit{#1}\fi}

\newcommand{\eqnlabel}[1]  {\label{eqn:#1}}
\newcommand{\eqnref}[1]    {Equation~\ref{eqn:#1}}

\newcommand{\lilabel}[1]   {\label{li:#1}}
\newcommand{\liref}[1]     {line~\ref{li:#1}}
\newcommand{\dist}[2][]
{\id{dist}\ifblank{#1}{}{_{#1}}\ifblank{#2}{}{^{#2}}}
\newcommand{\reach}[2][]
{R\ifblank{#1}{}{_{#1}}\ifblank{#2}{}{^{#2}}}

\newcommand{\thru}[2][]
{\id{thru}\ifblank{#1}{}{_{#1}}\ifblank{#2}{}{^{#2}}}
\newcommand{\bw}[2][]
{\id{BW}\ifblank{#1}{}{_{#1}}\ifblank{#2}{}{^{#2}}}

\newcommand{\card}[1]      {\left| #1\right|}
\newcommand{\ang}[1]       {\left<#1\right>}
\newcommand{\set}[1]       {\left\{#1\right\}}

\newcommand{\ceil}[1]      {\lceil#1\rceil}

\newtheorem{theorem}{Theorem}[section]
\newtheorem{lemma}[theorem]{Lemma}
\newtheorem{corollary}[theorem]{Corollary}

\newtheorem{definition}[theorem]{Definition}

\newenvironment{closeitemize}
    {\begin{list}{$\bullet$}
    {
         \setlength{\itemsep}{-0.3\baselineskip}
     \setlength{\topsep}{0.15\baselineskip}
     \setlength{\parskip}{0pt}}}
    {\end{list}}

\newcounter{ccount}
\newenvironment{closeenum}
    {\begin{list}{\arabic{ccount}.}
    {\usecounter{ccount}\setlength{\itemsep}{-0.3\baselineskip}
     \setlength{\topsep}{0.15\baselineskip}
     \setlength{\parskip}{0pt}}}
    {\end{list}}

\begin{document}
\begin{titlepage}
  \title{Single-Source Shortest Paths with \\Negative Real Weights in
    $\Otilde{mn^{8/9}}$ Time}
  \author{Jeremy T.\ Fineman\\
    Georgetown University\\
    \texttt{jf474@georgetown.edu}} \date{}
  \maketitle

  \begin{abstract}
    This paper presents a randomized algorithm for the problem of
    single-source shortest paths on directed graphs with real (both
    positive and negative) edge weights. Given an input graph with $n$
    vertices and $m$ edges, the algorithm completes in
    $\Otilde{mn^{8/9}}$ time with high probability.  For real-weighted
    graphs, this result constitutes the first asymptotic improvement
    over the classic $O(mn)$-time algorithm variously attributed to
    Shimbel, Bellman, Ford, and Moore.  
  \end{abstract}
  \thispagestyle{empty}
\end{titlepage}

\input{intro}

\input{prelim}

\input{overview}
\input{between}
\input{findsandwich}

\input{usesandwich}
\input{hopreduction}

\bibliographystyle{plain}
\bibliography{main}

\end{document}

%% file: intro.tex
\section{Introduction}\seclabel{intro}

This paper considers the problem of single-source shortest paths
(SSSP) with possibly negative real weights.  The input to the SSSP
problem is a directed graph $G=(V,E,w)$ with real edge weights given
by the function
$w: E \rightarrow \mathbb{R}$ and a designated source vertex $s$.  If the
graph does not contain any negative-weight cycles, then the goal is to
output the shortest-path distance from the source $s$ to every vertex
$v\in V$.  If there is a negative-weight
cycle in the graph, then the algorithm should instead report the
presence of such a cycle.\footnote{Some algorithms 
  only report a negative-weight cycle if such a cycle is reachable from~$s$.
  But it is not hard to build black-box reductions from each version of
  the problem to the other.}

The classic algorithm for SSSP with real weights, due to
Shimbel~\cite{Shimbel55}, Ford~\cite{Ford56},
Bellman~\cite{Bellman58}, and Moore~\cite{Moore59}, henceforth called
the Bellman-Ford algorithm, has a running time of $O(mn)$ on a graph
with $m$ edges and $n$ vertices.  With no further restrictions to
graph topology or weights, this algorithm remains the best known
algorithm for SSSP.  When the weights are all \emph{nonnegative}
reals, Dijkstra's
algorithm applies, which can be made to run in $O(m+n\log n)$
time~\cite{FredmanTa87}.

For the case of \emph{integer} weights (negative and positive), there has
been significant further progress~\cite{GabowTa89, Goldberg95,
  CohenMaSa17, AxiotisMaVl20, BrandLeNa20, ChenKyLi22}, culminating in
a nearly linear-time algorithm~\cite{BernsteinNaWu22}.  All of these
integer-weight solutions apply a scaling approach, and their running
times depend on at least a $\log W$ term, where $-W$ is the
most-negative weight in the graph.  The $O(mn)$-time Bellman-Ford algorithm
remains the best strongly polynomial runtime known even for
the case of integer weights.

The main result of this paper is captured by the following
theorem. 
Throughout the paper, the model used is a Real RAM to allow for 
standard manipulation of edge weights in constant time; specifically,
addition, subtraction, negation, and comparison of real numbers (i.e., the
edge weights) each take constant time.  (The algorithms
presented in this paper do not perform any multiplication or division
of edge weights.) 

\begin{theorem}\thmlabel{mainthm}
  There exists a (Las Vegas) randomized algorithm that solves the SSSP
  problem for real-weighted graphs in $\tilde{O}(m n^{8/9})$ time, with high probability,
  where $m$ is the number of edges and $n$ is the number of vertices
  in the graph.
\end{theorem}


\subsection{Preliminaries}
The $\tilde{O}$ denotes the soft-O notation.  Formally,
$f(x) = \tilde{O}(g(x))$ if there exists an integer $k$ such that
$f(x) = O(g(x) \cdot \log^k(g(x)))$. 
  
For the following, consider a graph $G=(V,E,w)$, let $m = \card{E}$
and $n = \card{V}$.  For a path $p$, the total \defn{weight} of the
path is given by $w(p) = \sum_{e \in p} w(e)$.  The \defn{size} of the
path is the number of edges on the path, denoted by $\card{p}$.  A
\defn{cycle} $C$ is a path that starts and ends at the same vertex,
and a negative-weight cycle is one where $w(C) < 0$.  
A
path $p$ from $u$ to $v$ is a \defn{shortest path} if all $u$-to-$v$
paths $p'$ satisfy $w(p) \leq w(p')$.  If there exists a shortest path
$p$ from $u$ to $v$, then we define the \defn{shortest-path distance}
from $u$-to-$v$ as $\dist[G]{}(u,v) = w(p)$; if there is no $u$-to-$v$
path, then $\dist[G]{}(u,v) = \infty$; if there is a path but no
shortest path (i.e., there is a negative-weight cycle), then
$\dist[G]{}(u,v) = -\infty$.  When $G$ is clear from context, we often
write $\dist{}(u,v)$ in place of $\dist[G]{}(u,v)$.

For a subset $S \subseteq V$ of vertices, the shortest-path distance
from any vertex in $S$ to $v$, denoted by 
$\dist[G]{}(S,v)$, is defined as 
\[ \dist[G]{}(S,v) = \min_{u \in S} \left(\dist[G]{}(u,v)\right) \ . \] The problem
of computing $\dist[G]{}(S,v)$ for all $v\in V$ corresponds to that of
solving SSSP on a slightly augmented graph:
create a ``super source'' vertex $s$, for all $u\in S$ add edges
$(s,u)$ with $w(s,u)=0$ to the graph, and finally solve SSSP from the
super source $s$ in the augmented graph.  Johnson's
algorithm~\cite{Johnson77} uses this same graph augmentation with
$S=V$.

\paragraph{Simplifying assumptions (without loss of generality).}
We shall make the following assumptions about the input graph
throughout. (1) If $(u,v) \in E$ and $w(u,v) < 0$, then $u$ has only
one outgoing edge;
thus, there are at most $n$ negative-weight edges in the
graph.\footnote{This first assumption is for convenience of
  exposition, not to simplify the algorithm.  The assumption implies a
  one-to-one correspondence between negative-weight edges and vertices
  with outgoing negative-weight edges, so referring to either is
  equivalent. Without the assumption, various statements and
  definitions would need to be altered, but the algorithm would
  otherwise remain unchanged.}  (2) Every vertex has degree at most
$O(m/n)$; thus, a subgraph on $n/r$ vertices has $O(m/r)$
edges.\footnote{This second assumption is common in randomized graph
  algorithms. Unlike the first, this one does simplify the
  algorithm. For example, to obtain the same results without this
  assumption, vertices would have to be randomly sampled proportional
  to their degree instead of uniformly.} These assumptions are without
loss of generality as they can be obtained from an arbitrary input
graph via a simple graph transformation without increasing the size of
the graph by more than a constant factor and without changing
distances between vertices in the original vertex set.

We shall also assume that $m\geq 2n$ to keep some of the statements of
performance bounds more concise.  A constant of at least two here also
implies that the number of edges is dominated by the number of edges
with nonnegative weight.

\paragraph{Hop-limited shortest paths.}  It is a simple exercise to
construct a SSSP algorithm that runs in $\tilde{O}(hm)$ time when
shortest paths are limited to $h\geq 1$ negative-weight edges or ``hops.''
(\secref{prelim} introduces corresponding notation and briefly summarizes 
such an algorithm.)  The novel algorithm in this paper applies
hop-limited SSSP as a subroutine.  

\subsubsection*{Price functions}

As with most of the integer-weight algorithms for SSSP, the algorithm
in this paper relies on price functions introduced by
Johnson~\cite{Johnson77} to transform the graph to an equivalent one
without negative weights; then Dijkstra's algorithm can be used to
solve the SSSP problem on the reweighted graph.  In more detail, a
\defn{price function} is a function $\phi : V \rightarrow \mathbb{R}$.
Given a price function $\phi$, define
$w_{\phi}(u,v) = w(u,v) + \phi(u) - \phi(v)$ and
$G_{\phi} = (V,E,w_\phi)$.  Modifying the weights in this way has the
following key properties~\cite{Johnson77}: (1) every cycle $C$ has the
same weight in both $G$ and $G_\phi$, so negative-weight cycles are
preserved, and (2) a path $p$ is a shortest path in $G_{\phi}$ if and
only if it is a shortest path in $G$.  More precisely, all $u$-to-$v$
paths $p$ satisfy $w_\phi(p) = w(p) + \phi(u) - \phi(v)$; if $p$ is a
cycle then $\phi(u) = \phi(v)$ and hence $w_\phi(p)=w(p)$.  Price functions also
compose in the natural way, i.e.,
$(w_{\phi_1})_{\phi_2}(u,v) = w_{\phi_1+\phi_2}(u,v)$. 

We call $\phi$ or $w_\phi$ a \defn{valid reweighting} if $w_\phi$ does
not cause any edge weights to become negative.  That is, if
$\forall e \in E ((w(e) \geq 0) \implies (w_\phi(e) \geq 0))$.  We
say that $\phi$ or $w_\phi$ \defn{eliminates} a negative edge
$e \in E$ if $w(e) < 0$ and $w_\phi(e) \geq 0$.

Johnson~\cite{Johnson77} shows that (assuming no negative-weight
cycles) the problem of eliminating 
\emph{all} negative-weight edges can be accomplished by setting
$\phi(v) = \dist{}(V,v)$.  Using Bellman-Ford to solve the super-source problem, the running time is $O(mn)$.   When there are $k\ll n$ negative-weight edges, applying
hop-limited SSSP is better, giving a running time of $\tilde{O}(km)$. 

\subsection{Main Result}

This paper solves the problem of efficiently computing a reweighting
that eliminates a significant number of negative-weight edges. We say
that an algorithm is an \defn{$f(k)$-elimination algorithm} if, when
given an input graph $G=(V,E,w)$ with
$k=\card{\set{e\in E| w(e) < 0}}$ negative-weight edges, the algorithm
either (1) computes a valid reweighting that eliminates at least
$f(k)$ negative-weight edges\footnote{The reweighted graph $G_\phi$
  thus has at most $k-f(k)$ negative-weight edges.}, or (2) correctly
determines that the graph contains a negative-weight cycle.  Given an
$f(k)$-elimination algorithm $\cal A$, SSSP can be solved by
repeatedly applying $\cal A$ until no negative-weight edges remain,
and then applying Dijkstra's algorithm.\footnote{To be useful, the
  running time of the elimination algorithm should be much better than
  $O(mf(k))$, i.e., much better than $O(m)$ per edge
  eliminated. Obtaining an algorithm whose runtime is
  $\Otilde{m/k^\epsilon}$ per eliminated edge would generally
  translate to an $\Otilde{mn^{1-\epsilon}}$ algorithm for SSSP.} This
strategy of gradually eliminating negative-weight edges is reminiscent
of Goldberg's algorithm~\cite{Goldberg95} for integer-weighted graphs.

\begin{theorem}\thmlabel{main-elim}
  There exists a randomized $\Theta(k^{1/3})$-elimination algorithm
  for real-weighted graphs 
  that has running time $\tilde{O}(mk^{2/9})$, with high
  probability, where $m$
  and $k$ are the number of edges and negative-weight
  edges in the input graph, respectively.
\end{theorem}

\paragraph{\thmref{mainthm} is a corollary of \thmref{main-elim}.}
A similar argument occurs in~\cite{Goldberg95}, so the full proof is
omitted here.  The main idea is that $O(k^{2/3})$ repetitions of
$\Theta(k^{1/3})$-elimination suffice to reduce the number of
negative-weight edges by a constant factor.  The total running time of
these repetitions is
$\tilde{O}(mk^{8/9}) = \tilde{O}(mn^{8/9})$ to reduce
$k\leq n$
by a constant factor.  And $O(\log n)$ of these constant-factor reductions are
enough to eliminate all negative-weight edges.

\paragraph{Sketch of algorithm.} The remainder of this paper focuses
on solving the problem of $\Theta(k^{1/3})$-elimination, thereby
proving \thmref{main-elim}.  At a very high level, the algorithm
entails reweighting the graph so that $\Omega(k^{1/3})$ of the
negative-weight edges are ``remote'' or ``far away'' from most of the
graph. (In particular, only an $O(1/k^{1/9})$ fraction of the graph is
``nearby'' these edges.)  Then, reweight the graph again to eliminate these
$\Theta(k^{1/3})$ negative-weight edges by applying Johnson's
strategy.  Because these remote edges are far from most of the graph,
it turns out that it is possible to eliminate these edges in
$\tilde{O}(k^{1/3}\cdot (m/k^{1/9})) = \tilde{O}(k^{2/9}m)$ time, which
improves over the straightforward but insufficient $\tilde{O}(k^{1/3}m)$
bound by a factor of~$k^{1/9}$.

A key challenge is, of course, to establish this remote subset of
negative-weight edges.  The algorithm modifies the starting graph in
two ways as it progresses: the algorithm performs several gradual
reweighting steps to ensure remoteness, and the algorithm drops some
negative edges from consideration.  Each gradual reweighting uses
hop-limited shortest paths.  In slightly more detail, the first
reweighting selects a random sample of vertices and uses hop-limited
shortest paths to ``spread out'' the graph.  Then, search for a large
subset of negative edges that are relatively ``close together,'' or
failing that 
find a large subset that are ``independent.''  (Resolving
the latter case is easier.) Drop all other negative edges
from consideration.  Another reweighting moves most of the graph
away from those close-together edges, making the edges remote. Then a
final reweighting step is performed to eliminate these now remote edges;
this last reweighting is the only one guaranteed to eliminate any
negative-weight edges.

\subsubsection*{Outline}

Before giving any further detail of the algorithm, \secref{prelim}
establishes useful notations and definitions to formalize these types
of manipulations.  \secref{overview} then gives an overview of the
algorithm with some intuition. Finally,
\secrefs{between}{hopreduction} provide details of each step of the
algorithm and the analysis.

%% file: prelim.tex
\section{Preliminaries}
\seclabel{prelim}

This section provides basic definitions and notation.  In addition,
this section discusses one of the main black-box subroutines:
hop-limited shortest path.  There are various definitions introduced
later in the paper as well, but most of those represent novel insights
into the structure of an efficient solution.  This section also
includes several useful claims for which the proofs are all simple
exercises and hence omitted. 

\paragraph{General graph notation.}
Consider a graph $G=(V,E,w)$, and let $X\subseteq V$ be any subset of
vertices.  Then $\id{out}(X)$ denotes the set of outgoing edges from
$X$, i.e., $\id{out}(X) = \set{(x,y) \in E | x \in X}$.

For a fixed target $t$, the problem of computing $\dist[G]{}(u,t)$ for
all $u\in V$ is called the \defn{single-target shortest-paths (STSP)}
problem.  This problem can be solved by solving SSSP from $t$ in the
transpose graph.  The \defn{transpose graph} is the graph obtained by
reversing all the edges.  That is, the transpose graph is a graph
$G^T = (V,E^T,w^T)$ where $E^T = \set{(v,u) | (u,v)\in E}$ and
$w^T(v,u) = w(u,v)$.

\paragraph{Negative edges, nonnegative edges, and the input graph.}
The \defn{input graph} refers to the graph $G$ on which the main
algorithm of \thmref{main-elim} is called, possibly with a modified
weight function.  We shall always denote the
input graph by $G = (V, E^+ \cup E^-, w)$, where the edge set has been
partitioned into the \defn{nonnegative edges}~$E^+$ and the
\defn{negative edges}~$E^-$.  Initially, $E^+ = \set{e \in E | w(e) \geq 0}$ and
$E^- = \set{e \in E | w(e) < 0}$, where $E=E^+\cup E^-$ is the full
edge set.  For every edge
$(u,v) \in E^-$, the vertex $u$ is called a \defn{negative vertex}.
Recall that, WLOG, every negative vertex has one outgoing edge.
Throughout, let $n=\card{V}$, $m=\card{E}$, and $k=\card{E^-}$.

As a slight abuse of notation, the $\cup$ symbol in
$G=(V,E^+ \cup E^-,w)$ is not simply a union, but also signifies which
edges are classified as negative edges (those in $E^-$), and which are
nonnegative (those in $E^+$).  As the algorithm progresses, the weight
function changes, but the classification of edges does not.  Thus,
having a negative edge $(u,v) \in E^-$ with $w(e) \geq 0$ is allowed;
that edge is still called a negative edge, and $u$ is still called
negative vertex.  In contrast, because the algorithm only produces
valid price function, it shall always be the case that $w(e) \geq 0$
for all $e\in E^+$.

Whenever the partition is not provided, e.g., if referring to an
auxiliary graph $H = (V',E',w')$, then implicitly the term ``negative
edges'' refers to those edges whose weight is negative.

\paragraph{(Negative)-hop-limited paths and distances.}

A path $p$ is an \defn{$h$-hop path} if at most $h$ of the
edges on the path are negative edges.  Nonnegative edges do not count
towards the number of hops. Paths need not be
simple, and each occurrence of a negative edge contributes to the hop
count.

The $h$-hop distance, denoted
\[ \dist[G]{h}(u,v) = \min\set{w(p) | \text{$p$ is an $h$-hop path
      from $u$ to $v$ in $G$}} \ , \] is the weight of a shortest
$h$-hop path from $u$ to $v$; define $\dist[G]{h}(u,v) = \infty$ if
there is no path from $u$ to $v$.  We also extend the distance
notation to be distance from a set of vertices (as in
\secref{intro}). Specifically, for any $S\subseteq V$, define
$\dist[G]{h}(S,v) = \min_{u\in S}\left(\dist[G]{h}(u,v)\right)$.
When $G$ is clear from context, we often write $\dist{h}$ instead of
$\dist[G]{h}$.
Note that unlike regular
distance, if $v$ is reachable from $u$, then $\dist[G]{h}(u,v)$ is
always finite, even with negative-weight cycles.

Just as with normal distance, it is easy to see that $h$-hop distances
obey a form of the triangle inequality, which has been adjusted to
incorporate the hop counts.
\begin{lemma}[Triangle inequality]
  For all integers $h_1,h_2\geq 0$ and all vertices $x, y, z$, we have
  \[ \dist{h_1+h_2}(x,z) \leq \dist{h_1}(x,y) + \dist{h_2}(y,z) \ . \]
  It follows that for any nonnegative edge $(y,z)$, $\dist{h_1}(x,z) \leq
  \dist{h_1}(x,y)+w(y,z)$.
  \lemlabel{triangle}
\end{lemma}

If $\dist[G]{h}(u,v) < 0$ or $\dist[G]{h}(v,u) < 0$, then we say that
$u$ and $v$ are \defn{$h$-hop related}.
The \defn{negative $h$-hop reach} of a vertex $u$ is the set of vertices
that can be reached by a negative-weight $h$-hop path.  More
generally, for a set subset $S \subseteq V$ of vertices, the negative $h$-hop
reach of $S$ is
\[ \reach[G]{h}(S) = \set{v \in V | \dist[G]{h}(S,v) < 0} \ . \]
The \defn{size} of the reach is its cardinality.  As with distance,
the subscript $G$ may be dropped when $G$ is clear from context.

\paragraph{Reweighting and invariance of $h$-hop paths.} 
The algorithm performs several steps that each partially reweight the
graph by way of a sequence of price functions~$\phi$.  The notation $G_\phi =
(V,E^+\cup E^-, w_\phi)$ denotes the reweighted graph, i.e., the input
graph reweighted by price function~$\phi$.  When $G$ is clear from context, we use the
subscript $\phi$ as a shorthand for $G_\phi$ in all notations where
the subscript specifies the graph of concern, i.e., $\dist[\phi]{h}$
means $\dist[G_\phi]{h}$ and $\reach[\phi]{h}$ means
$\reach[G_\phi]{h}$.  

The classification of edges as negative or nonnegative does not change when the graph is
reweighted, and the validity of the price function is defined with
respect to the initial classification.
Specifically, a price function $\phi$ is \defn{valid} if for all
$e \in E^+$, $w_\phi(e) \geq 0$.  When going from a price function
$\phi$ to a price function $\phi'$, the function $\phi'$ can still be
valid even if there exists $e \in E^-$ with $w_\phi(e) \geq 0$ and
$w_{\phi'}(e) < 0$.

Importantly, since the classification of edges does not change, $h$-hop paths in the input graph are invariant across
reweighting. That is, a path $p$ is an $h$-hop path in
$G_\phi=(V,E^+\cup E^-,w_\phi)$ if and only if it is an $h$-hop path
in~$G=(V,E^+\cup E^-,w)$.  Ensuring this
invariant is the primary reason negative edges were defined in the
specific manner above.  This invariant shall allow us to more-cleanly
reason about paths and distances when the algorithm performs several
reweighting steps. Specifically, we immediately have the
following.

\begin{lemma}\lemlabel{phidist}
  Consider the input graph $G=(V,E^+ \cup E^-,w)$, and let $\phi$ be a
  price function.  Then for all $u,v \in
  V$, we have
  \[ \dist[\phi]{h}(u,v) = \dist{h}(u,v) + \phi(u) - \phi(v) \ .\]
\end{lemma}

\paragraph{Computing $h$-hop distances.}
Given a source vertex $s$, the problem of computing $h$-hop distances
from $s$ to all other vertices is called the \defn{$h$-limited SSSP
  problem}.  There is a natural solution for $h$-limited SSSP that
combines Bellman-Ford and Dijkstra's algorithm, called BFD
here.\footnote{See, e.g.,~\cite{DinitzIt17}, for a deeper discussion
  of one variant of this algorithm.  Bernstein et
  al.~\cite{BernsteinNaWu22} apply an optimized version of BFD that
  does not reconsider a vertex in the next round unless its distance
  has improved; their algorithm for integer-weight SSSP leverages a tighter bound for the case that most
  shortest paths have few hops.} BFD interleaves $(h+1)$ full
executions of Dijkstra's algorithm (but without reinitializing
distances) on the nonnegative edges and $h$ ``rounds'' of Bellman-Ford
on the negative edges.\footnote{A ``round'' of Bellman-Ford means
  ``relaxing'' all the edges once. A full execution of Bellman-Ford is
  $n$ rounds.}  The running time of BFD is thus $O(hm\log n)$ when
$h\geq 1$.

\begin{lemma}[Follows from, e.g.,~\cite{BernsteinNaWu22,DinitzIt17}] Consider a graph
  $G=(V,E^+\cup E^-,w)$ with $w(e)\geq 0$ for all $e\in E^+$, and let
  $n=\card{V}$, $m=\card{V}$, and $k=\card{E^-}$. BFD solves the
  $h$-limited SSSP problem in time $O((h+1)(m+n\log n))$.  That is,
  given source vertex $s$ and integer $h\geq 0$, it returns 
  $d_h(v) = \dist[G]{h}(s,v)$ for all $v \in V$.  Moreover, the
  algorithm can also return all smaller-hop distances
  $d_{h'}(v) = \dist[G]{h'}(s,v)$ for all $h' \in \set{0,1,2,\ldots,h}$ with the
  same running time.

  When $h=k\geq 1$, BFD solves the regular SSSP problem in
  $O(k(m+n\log n))$ time.

  More generally, given a set $S\subseteq V$ instead of a source
  vertex, it is also possible to compute the distances
  $d_{h'}(v) = \dist[G]{h'}(S,v)$ for all $v\in V$ and $h' \leq h$
  with the same time complexity.  In addition, for all $v\in V$, the
  algorithm can be augmented to return $s(v) \in S$ such that $d_h(v) =
  \dist[G]{h}(s(v),v)$.\lemlabel{BFD}
\end{lemma}
  
Note that many textbook descriptions of Bellman-Ford (e.g.,
CLRS~\cite{CLRS}) update distance estimates in place, which when
extended to BFD would only guarantee $d_h(v)
\leq\dist[G]{h}(s,v)$.  The inequality may be problematic when
reasoning about hop-limited paths. We instead want the return values
to be exactly equal to the $h$-hop distances.  BFD of the lemma thus
starts from a version of Bellman-Ford that explicitly stores distances
for each round (e.g., Kleinberg-Tardos~\cite{KleinbergTa06}).

\paragraph{Subgraphs of negative edges.}  
For a subset $N \subseteq E^-$ of negative edges on the input graph,
we use $G^{N}$ to denote the subgraph $G^N = (V,E^+\cup N,w)$.
Moreover, $G^N_\phi$ denotes the reweighted subgraph
$G^N_\phi = (V, E^+\cup N, w_\phi)$.  Vertices are classified as
negative vertices in $G^N$ only if their corresponding negative edge
is included in~$N$.

Because all of the nonnegative edges are included in $G^N$, it should
be obvious that for any price function $\phi$, if $w_\phi$ is a valid
reweighting of $G^N$ then it is also a valid reweighting of $G$.
Moreover, if $w_\phi$ eliminates negative edges from $G^N$, it also
eliminates those same edges from $G$.  Working with subgraphs $G^N$
thus suffices to solve the problem.  Specifically, the algorithm shall
eventually reach a subgraph $G^N$ with $\card{N} = \Theta(k^{1/3})$
and find a reweighting that eliminates all the edges~$N$.

%% file: overview.tex
\section{Algorithm Overview}
\seclabel{overview}

This sections provides an overview of the algorithm for
$\Theta(k^{1/3})$ elimination.  This section includes some intuition
of correctness for each of the main components of the algorithm, but
the full details and most of the proofs are deferred to \secrefs{between}{hopreduction}.

The main goal of the algorithm is to find a large (i.e.,
size-$\Theta(k^{1/3})$) $r$-remote set or a
large $1$-hop independent set, both defined next, and then to
eliminate the corresponding negative edges.  (We shall eventually set
$r=\Theta(k^{1/9})$). 

\begin{definition}
  Consider a graph $G=(V,E^+\cup E^-,w)$, let
  $n=\card{V}$, and let $X$ be a subset of negative vertices.
  If the negative $r$-hop reach of $X$ has size at most $n/r$, i.e.,
  $\card{\reach{r}(X)} \leq n/r$, then $X$ is an \defn{$r$-remote
    set}.  We also call $\id{out}(X)$ a set of \defn{$r$-remote edges},
  and we call the subgraph induced by the negative $r$-hop reach of $X$ an
  \defn{$r$-remote subgraph}.
\end{definition}

\begin{definition}
  Consider a graph $G=(V,E^+\cup E^-,w)$.  Let $I$ be a subset of
  negative vertices.  We say that $I$ is a \defn{$1$-hop independent
    set} if $\forall x, y \in I$, $x$ and $y$ are not $1$-hop related
  in~$G$.
\end{definition}

\algref{full} outlines the algorithm.  Note that some of the
terminology will be revealed later in this section.  Nevertheless, the
reader may wish to refer to this psuedocode to see how the steps fit
together.  Each of the main steps is marked with the corresponding
sections that explain them.  For expository reasons, the steps of the
algorithm are presented out of order in this overview section (but in
order later in the paper).  The algorithm produces a sequence of price
functions through several steps.  Each step computes the next price
function relative to the current weighting of the graph.  Thus, the
actual weight is obtained by composing (adding) all of the price functions.

Roughly speaking, there are two main components in the algorithm.  The
first component is an efficient algorithm to either find a large
$r$-remote set (also with large $r$) or, failing that, to find a large
$1$-hop independent set.  Unfortunately, neither may exist with the
original weight function of the graph---it is not hard to construct
graphs where (1) every pair of negative vertices is $1$-hop related,
and (2) every negative vertex has large $1$-hop reach, i.e.,
$\card{\reach{1}(\set{u})} = \Omega(n)$.  The first component of the
algorithm thus entails not simply finding such a set, but also
adjusting the weight function to ensure that such a set exists.  This
component spans all but the last numbered step in the pseudocode.

The second component is an efficient algorithm that eliminates all of
outgoing edges from the $r$-remote or $1$-hop-independent
set.  The second problem is easier, and it also helps
to motivate why $r$-remote sets are useful. Thus, this section
addresses the second component first.  (Efficiently eliminating a $1$-hop
independent set is almost trivial, so that is deferred to \secref{combining}.)

\begin{algorithm}[t]
\caption{Algorithm for eliminating $\Theta(k^{1/3})$ negative
  edges. Negative-weight cycles may be discovered inside steps (1), (2),
or (4); when a cycle is discovered, the entire algorithm is terminated.}\alglabel{full}
\SetKwInOut{Input}{input}
\Input{A graph $G=(V,E^+\cup E^-,w)$ with $w(e) \geq 0$ for $e\in E^+$
and $w(e) < 0$ for $e \in E^-$}
\ \\
let $k = \card{E^-}$ and $r= \Theta(k^{1/9})$\\
\nl (\secref{over-btw},~\ref{sec:between}) perform betweenness reduction on $G$ with $\beta=r+1$ and $\tau=r$\\
    let $\phi_1$ be the price function computed by this step\\
\nl (\secref{over-findsandwich},~\ref{sec:findsandwich}) find a size-$\Omega(k^{1/3})$ negative sandwich $(x,U,y)$ or independent set $I$
in $G_{\phi_1}$ \\
\If{this step discovers an independent set}
{(\secref{combining}) find a
    price function $\phi$ that eliminates all
    negative edges in $G_{\phi_1}^{\id{out}(I)}$\\\KwRet{$\phi+\phi_1$}}
  \lElse{arbitrarily remove vertices from $U$
    until $\card{U} = \Theta(k^{1/3})$}
\nl (\secref{over-usesandwich},~\ref{sec:usesandwich}) reweight the graph
$G_{\phi_1}$ to try to  make $U$ become $r$-remote\\
   let $\phi_2$ be the price function computed by this step\\
    \lIf{$\card{\reach[\phi_1+\phi_2]{r}(U)} > n/r$ (i.e.,
      not $r$-remote)}{restart \algref{full}}
\nl (\secref{over-hopreduction},~\ref{sec:hopreduction}) use the hop-reduction technique on graph
$G_{\phi_1+\phi_2}^{\id{out}(U)}$ to eliminate $\id{out}(U)$\\
let $\phi$ be the price function computed by this step\\
  \KwRet{$\phi+\phi_1+\phi_2$}
\end{algorithm}

\subsection{Eliminating $r$-remote edges by hop reduction}\seclabel{over-hopreduction}

Recall that Johnson's strategy~\cite{Johnson77} for eliminating
negative edges entails solving SSSP.  If there are $\hat{k}$ negative
edges, then the running time is $\tilde{O}(m\hat{k})$ using BFD.  The
goal here is to accelerate this SSSP computation for the case that the
edges being eliminated are remote. (Notationally, the use of $\hat{k}$
here is to emphasize that this step is applied to a subgraph.)

To illustrate the approach, consider first an arbitrary graph
$G=(V,E^+\cup N,w)$ without a known remote set.  
The goal is
to produce a new auxiliary graph $H = (V_H, E_H, w_H)$ such (1)
$V \subseteq V_H$, and (2) for all hop counts $h\geq 0$ and
$u,v \in V$,
$\dist[G]{h}(u,v) \geq \dist[H]{\ceil{h/r}}(u,v) \geq
\dist[G]{}(u,v)$.  That is to say, all $h$-hop paths in $G$ correspond
to $\ceil{h/r}$-hop paths in~$H$.  We say that $H$ is an \defn{$r$-hop
  reduction of $G$}.  Thus, we can compute SSSP for $G$ by instead
computing SSSP in $H$ with a cost of
$\tilde{O}((\hat{k}/r) \cdot m_H)$, where $\hat{k} = \card{N}$ is
the number of negative edges in $G$, and $m_H = \card{E_H}$ is the size
of $H$.  As we shall see next, there is a fairly straightforward
construction of an $O(rm)$-size $r$-hop reduction of $G$.
Unfortunately, the running time of SSSP remains
$\tilde{O}((\hat{k}/r)\cdot (rm)) = \tilde{O}(\hat{k}m)$.  But given an
$r$-remote set, it is possible to improve this construction to achieve
a better running time.

The construction of $H$ is roughly as follows.  First, for each vertex
$v \in V$, add $r+1$ copies $v=v_0, v_1, \ldots, v_r$ to $V_H$.  Add
the nonnegative edges to each layer of the graph, i.e., for each edge
$(u,v) \in E^+$ and each $0\leq i \leq r$, add the edge $(u_i,v_i)$ to
$E_H$.  As for the negative edges $(u,v)\in N$, create the edges
$(u_i,v_{i+1})$ for $0\leq i <r$ to $E_H$.  Each copy of the negative
edge thus moves from the $i$-th layer of the graph to the $(i+1)$th
layer.  Finally, add edges $(v_i,v_0)$ for all $v$ and~$i$ to allow a
way to get back to the 0th layer.

It remains to specify the weight function $w_H$.  The goal is to
ensure that only the edges $(v_i,v_0)$ have negative weight, and thus
an $r$-hop subpath in $G$ can be simulated by a $1$-hop path in $H$
that moves through copies $0, 1, 2,\ldots r, 0$.  
This goal can be accomplished by roughly
running Johnson's reweighting limited to $r$ hops, i.e., computing $i$-hop SSSP from $V$ for all $i\leq r$,
and setting $w_H(u_i,v_j) = w(u,v) + \dist[G]{i}(V,u) - \dist[G]{j}(V,v)$.
For each $(u,v) \in N$, it follows that $w_H(u_i,v_{i+1}) \geq 0$ because
$\dist[G]{i+1}(V,v) \leq \dist[G]{i}(V,u) + w(u,v)$.

The graph $H$ has size $m_H = O(r m)$ by construction.  Moreover, from
\lemref{BFD} the SSSP distances and hence weights $w_H$ can be
calculated in $\tilde{O}(rm)$ time.

Now let us improve the construction if $N$ is a set of $r$-remote
edges.  Consider a vertex $u \in V$ that falls outside the $r$-remote
subgraph.  Then $r$-remoteness implies that $\dist[G]{i}(V,u) = 0$ for
all $i\leq r$ as there is no negative-weight path and there is a
0-weight path (the empty path from $u$).  There is thus no reason to
include multiple copies of this vertex in $H$ as each copy's incident
edges would be weighted identically---it suffices to keep the single
copy $u = u_0$, or equivalently to contract all copies into $u$ and
remove any redundant edges. In summary, when given an $r$-remote
subgraph, $H$ comprises $r$ copies of the remote subgraph plus a
single copy of the original graph.  Applying the assumption that the
maximum degree is $O(m/n)$, the total size of $H$ now becomes
$m_H = O(r\cdot (n/r)\cdot (m/n) + m) = O(m)$. Moreover, $H$ still
constitutes an $r$-hop reduction of $G$.  We are thus left with the
following lemma; (the second term in the runtime is the cost of
constructing~$w_H$).

\begin{lemma}\lemlabel{hop-reduction}
  Consider a graph $G=(V,E^+\cup N, w)$; let
  $\hat{k} = \card{N}$ and $m=\card{E^+\cup N}$.  Suppose that
  $w(e) \geq 0$ for all $e \in E^+$ and that $N$ is $r$-remote.
  Then there exists an $\tilde{O}((\hat{k}/r) m + rm)$-time
  deterministic algorithm that either (1) correctly determines that the
  graph contains a negative-weight cycle, or (2) computes a valid
  reweighting that eliminates all edges $N$.
\end{lemma}

\subsection{Betweenness reduction}\seclabel{over-btw}

We are left with the more difficult problem of uncovering an
$r$-remote set or $1$-hop independent set, which as previously noted
entails some reweighting. But it is not clear how to attack this
problem directly.  Roughly speaking, the challenge is that when the
price of a vertex changes, there may be new $h$-hop relationships
introduced even though no negative edges are created.  It thus seems
difficult to argue that a particular reweighting of the graph reduces
the number of relationships. 

The key insight here is to think in
terms of ``betweenness'' instead,  which is better behaved.  We can
then later translate to an $r$-remote set, but that transformation
is more restricted so easier to reason about. 

\begin{definition}
  For the following, consider a graph $G$, vertices $u$,
  $x$, and $v$, and integer $\beta \geq 0$.

  The \defn{$\beta$-distance from $u$ to $v$ through $x$} is defined
  as
  \[ \thru[G]{\beta}(u,x,v) = \dist[G]{\beta}(u,x) +
    \dist[G]{\beta}(x,v) \ . \]
  We say that $x$ is \defn{$\beta$-between}
  $u$ and $v$ if $\thru[G]{\beta}(u,x,v) < 0$.  The
  \defn{$\beta$-betweenness} of $u$ and $v$, denoted
  $\bw[G]{\beta}(u,v) = \card{\set{x \in V| \thru[G]{\beta}(u,x,v) < 0}}$
  is the number of vertices $\beta$-between $u$ and $v$.
\end{definition}
\noindent For all of these notations, the $G$ may be dropped if clear
from context, and $\phi$ is used as shorthand for $G_\phi$.

The goal here is to find a price function
$\phi$ so that for given parameter $\tau$, all pairs $u,v \in V$ have
$\bw[\phi]{\beta}(u,v) \leq n/\tau$.  (We will use
$\tau = \beta-1 = r$, but the algorithm of this section is described
for any $\beta$ and $\tau$.) The algorithm
is fairly simple. Sample a size-$\Theta(\tau \log n)$ subset of
vertices.  Then find any reweighting for which all
$\beta$-hop distances to or from the sampled vertices are nonnegative,
or determine that the graph contains a negative-weight cycle.
Roughly speaking, the reweighting entails computing
$\Theta(\beta \tau \log n)$-limited SSSP (because we want
$\Theta(\beta)$-hop subpaths between each of the $\Theta(\tau \log n)$
samples).  There are many relatively straightforward ways to achieve
the desired reweighting, and the details are deferred to \secref{between}.

We are left with a question: does reweighting in this way ensure that
$\bw[\phi]{\beta}(u,v) \leq n/\tau$?
It is easy to see that (by construction) no sampled vertex is
$\beta$-between any pair of vertices in $G_\phi$, but that would only directly
tell us that the $\beta$-betweenness is at most
$n-\Theta(\tau\log n)$.

Consider the distance from $u$ to $v$ through a vertex $x$.  It
follows from \lemref{phidist} that 
$\thru[\phi]{\beta}(u,x,v) = \thru{\beta}(u,x,v) + \phi(u) - \phi(v)$,
which importantly does not depend on $\phi(x)$.  The
$u$-to-$v$ distances through other vertices thus compare in the same
way before and after reweighting. Therefore, if any
sampled vertex $y$ has $\thru{\beta}(u,y,v) \leq \thru{\beta}(u,x,v)$,
then it follows that $x$ is not $\beta$-between $u$ and $v$ in
$G_\phi$ because $y$ is not either.  With high probability, there is a
sample $y$ taken from the smallest $1/\tau$-fraction of through
distances, and hence at most a $1/\tau$ fraction of vertices is
$\beta$-between $u$ and~$v$ in $G_\phi$. Thus, we obtain the following, with
proof in \secref{between}:

\begin{lemma}\lemlabel{betweenness-reduction}
  Consider input graph $G=(V,E^+\cup E^-,w)$; let
  $m=\card{E^+\cup E^-}$, and suppose that $w(e) \geq 0$ for all
  $e \in E^+$.  Then there exists an
  $\tilde{O}(\beta\tau m + \tau^2 n)$-time (Monte Carlo) randomized
  algorithm that always satisfies one of the following three cases,
  and it falls in either of the first two with high probability:
  (1) it correctly determines the graph contains a
  negative-weight cycle, (2) it finds valid price function
  $\phi$ such that $\bw[\phi]{\beta}(u,v) \leq n/\tau$ for all $u,v\in
  V$, or (3) it returns a valid price function, but the betweenness goal
  is not achieved.
\end{lemma}

\subsection{A sandwich with low betweenness gives $r$-remoteness}\seclabel{over-usesandwich}

Consider a graph $G=(V, E^+ \cup E^-, w)$.  The goal is to argue
that if $G$ has low betweenness, then it is not too hard to reweight
$G$ so that there is an $r$-remote subset.  To do so, we apply a new
object called a negative sandwich.

\begin{definition}A \defn{negative sandwich} is a
triple $(x, U, y)$ with the following properties.  
\begin{closeitemize}
\item $U$ is a subset of negative vertices,
\item $x\in V$ and $\dist{1}(x,u) < 0$ for all $u\in U$, and
\item $y\in V$ and $\dist{1}(u,y) < 0$ for all $u\in U$.
\end{closeitemize}
The \defn{size} of
the sandwich is the cardinality of $U$.
\end{definition}

For now, let us ignore the task of finding such a sandwich. The goal
here is only to argue that a negative sandwich is useful.  Note that
restricting $U$ to negative vertices does not affect the bulk of the
logic here
(\lemref{sandwich} holds for any set $U\subseteq V$); this restriction
is simply because negative vertices are what matter for the
transformation to a remote subset.

Given a negative sandwich $(x,U,y)$ and hop count $\beta$, consider the
reweighting given by the price function
$\phi(v) = \min(0,\max(\dist{\beta}(x,v),-\dist{\beta}(v,y)))$.
Roughly speaking, there are two main goals of this price function: (1)
for all $u \in U$, $\phi(u) = 0$, and (2) for most other vertices~$v$,
$\phi(v) \leq 0$ and 
$\phi(v) \leq \dist{\beta}(x,v)$.  Because the 1-hop distance from
$x$ to $u$ is negative (by definition of a negative sandwich), these together would imply that the
$(\beta-1)$-hop distance from $u$ to $v$ in the reweighted graph
becomes positive.  In general, however, ensuring (1) in a way that also gives a
valid reweighting somewhat interferes with (2).  This is why the
price function here uses $\dist{\beta}(v,y)$ to limit how negative
$\phi(v)$ can get. It is not hard to see that (1) is ensured because
in a negative sandwich
$\dist{\beta}(u,y) < 0$ for all $u\in U$.  The price function also
ensures (2) because when $v$ is \emph{not} $\beta$-between $x$ and
$y$, $\dist{\beta}(x,v) + \dist{\beta}(v,y) \geq 0$ or
$\dist{\beta}(x,v) \geq -\dist{\beta}(v,y)$; the implication is that
$\max(\dist{\beta}(x,v),-\dist{\beta}(v,y)) = \dist{\beta}(x,v)$ as desired.  That is to say,
the only vertices that remain in the $(\beta-1)$-hop reach of $U$ in
$G_\phi$ are (a subset of) those vertices that are $\beta$-between $x$
and $y$ in~$G$.  It follows that if $x$ and $y$ have
$\beta$-betweenness at most $n/\tau$, then $U$ becomes
$\min(\tau,\beta-1)$-remote.

The following lemma formalizes these ideas and also proves that the
reweighting is valid. That the reweighting is valid may not be
obvious, but the proof (in \secref{usesandwich}) essentially amounts
to applying the triangle inequality.

\begin{lemma}\lemlabel{sandwich}
  Consider a graph $G=(V,E^+\cup E^-, w)$, and suppose that
  $w(e) \geq 0$ for all $e \in E^+$.  Consider negative
  sandwich $(x,U,y)$ and integer $\beta > 1$. Let $\phi$
  be the price function defined by
  \[ \phi(v) =
    \min(0,\max(\dist[G]{\beta}(x,v),-\dist[G]{\beta}(v,y))) \ . \]
  Then we have the following:
  \begin{closeenum}
  \item $\phi$ is a valid reweighting, i.e., $w_\phi(e) \geq 0$ for
    all $e \in E^+$.
  \item For every $v\in V$: if $\thru[G]{\beta}(x,v,y) \geq 0$ (i.e.,
    $v$ is not $\beta$-between $x$ and $y$), then
    $v \not\in \reach[G_\phi]{\beta-1}(U)$.
  \end{closeenum}
\end{lemma}

We conclude with the following.

\begin{lemma}
  Consider a graph $G=(V,E^+ \cup E^-,w)$ with $w(e)\geq 0$ for all
  $e\in E^+$, and let $m=\card{E^+ \cup E^-}$ and $n=\card{V}$.
  Consider also a negative sandwich $(x,U,y)$ and any integer
  $\beta > 1$. Let $b=\bw{\beta}(x,y)$ denote the $\beta$-betweenness of $x$
  and $y$.

  Then there exists an $O(\beta m \log n))$-time deterministic algorithm that
  finds a valid reweighting $\phi$ such that $U$ is
  $\min(\beta-1,n/b)$-remote in $G_\phi$
\end{lemma}
\begin{proof}
  The qualities of $\phi$ follow from \lemref{sandwich} and the
  definition of $r$-remote.  The price function $\phi$ can be computed
  by solving $\beta$-limited SSSP from $x$ and $\beta$-limited STSP to
  $y$.   Applying the running time for BFD (\lemref{BFD}) completes
  the proof. 
\end{proof}

Choosing $\beta=r+1$ and $\tau=r$ as parameters in the betweenness
reduction (i.e., \lemref{betweenness-reduction}), we obtain the following
corollary:
\begin{corollary}\corlabel{sandwich}
  Suppose we are
  given a negative sandwich $(x,U,y)$ and that $\bw{r+1}(x,y) \leq
  n/r$ for integer $r \geq 1$. Then there is an
  $O(rm\log n)$-time deterministic algorithm that finds a valid price function $\phi$ such that $U$
  is $r$-remote in $G_\phi$.\qed
\end{corollary}

\noindent This step may fail to make $U$ become $r$-remote only if the Monte
Carlo betweenness reduction failed to ensure that $x$ and $y$ have low
betweenness; in this case, the entire
algorithm must be restarted.

\subsection{Finding a negative sandwich or independent set}\seclabel{over-findsandwich}

The final problem is that of finding a negative sandwich or
$1$-hop independent set.  The main tool is given by the following
lemma, proved in \secref{findsandwich}.

\begin{lemma}\lemlabel{sandwich-crust}
  Consider an input graph $G=(V,E^+\cup E^-, w)$ with $w(e) \geq 0$
  for all $e\in E^+$; let $n=\card{V}$, $m=\card{E^+\cup E^-}$.  Let
  $U_0$ be any subset of negative vertices in $G$, let $\hat{k} =
  \card{U_0}$, and let $\rho$ be an integer parameter with $1\leq \rho
  \leq \hat{k}$.  

  There exists a (Las Vegas) randomized algorithm whose running time
  is $O(m\log^2 n)$, with high probability, that takes as input $G$,
  $U_0$, and $\rho$ and always does one of the following:
  \begin{closeenum}
  \item correctly determines that $G$ contains a negative-weight cycle,
  \item returns a subset of negative vertices $U \subseteq U_0$ with $\card{U} =
    \Omega(\hat{k}/\rho)$ and a vertex $y$ such that
    for all $u\in U$, $\dist{1}(u,y) < 0$, or
  \item returns a $1$-hop independent set $I$ with $\card{I} = \Omega(\rho)$.
  \end{closeenum}
\end{lemma}

Given \lemref{sandwich-crust}, we immediately obtain the following
lemma by running the algorithm twice.

\begin{corollary}\corlabel{findsandwich}
  Consider an input graph $G=(V,E^+\cup E^-,w)$ with $w(e)\geq 0$ for
  all $e \in E^+$; let $n=\card{V}$, $m=\card{E^+\cup E^-}$, and $k =
  \card{E^-}$.   

  There exists a (Las Vegas) randomized algorithm whose running time
  is $O(m\log^2n)$, with high probability, that always does one of the
  following:
  \begin{closeenum}
  \item correctly determines that $G$ contains a negative-weight cycle,
  \item returns negative sandwich $(x,U,y)$ with $\card{U} =
    \Omega(k^{1/3})$, or
  \item returns a 1-hop independent set $I$ with $\card{I} =
    \Omega(k^{1/3})$.
  \end{closeenum}
\end{corollary}
\begin{proof}
  First, run the algorithm of \lemref{sandwich-crust} on $G$ with
  $U_0 = E^-$ and $\rho = \ceil{k^{1/3}}$.  From
  \lemref{sandwich-crust}, the algorithm does one of the following:
  (i) correctly determines that $G$ contains a negative-weight cycle,
  (ii) returns a 1-hop independent set with size
  $\Omega(\rho) = \Omega(k^{1/3})$, or (iii) returns a vertex $y$ and
  subset $U_1 \subseteq U_0$ with
  $\card{U_1} = \Omega(k/\rho) = \Omega(k^{2/3})$ such that
  $\dist{1}(u,y) < 0$ for all $u\in U_1$.  In cases (i) and (ii), we
  are done.  Otherwise, run the algorithm of \lemref{sandwich-crust}
  again but in the transpose graph (with all edges reversed) using
  $G$, $U_1$, and $\rho = \ceil{k^{1/3}}$.  The second execution thus either
  (i) identifies a cycle, (ii) returns an independent set with size
  $\Omega(\rho) = \Omega(k^{1/3})$, or (iii) returns a vertex $x$ and
  subset $U_2 \subseteq U_1$ with
  $\card{U_2} = \Omega(\card{U_1}/\rho) = \Omega(k^{2/3}/k^{1/3}) =
  \Omega(k^{1/3})$ such that $\dist{1}(x,u) < 0$ for all $u\in U_2$.
  In cases (i) and (ii), we are again done. In case (iii), $(x,U_2,y)$
  is a negative sandwich with size $\Omega(k^{1/3})$.
\end{proof}

We are left with the problem of constructively proving
\lemref{sandwich-crust}. For
the following, let
$C(U_0,v) = \card{\set{u \in U_0 | \dist[G]{1}(u,v) < 0}}$ denote the
number of vertices in $U_0$ that can reach $v$ with a negative-weight
$1$-hop path.

The algorithm for \lemref{sandwich-crust} is roughly as follows, with
details in \secref{findsandwich}.  The
first task is to estimate $C(U_0,v)$ for all $v\in V$.  More
precisely, the goal is to partition $U_0$ into two subsets $H$ and $L$
(for heavy and light, respectively), such that $\forall v\in H,C(U_0,v) =
\Omega(\hat{k}/\rho)$ and $\forall v \in L, C(U_0,v) = O(\hat{k}/\rho)$.  This
task can be accomplished by randomly sampling each vertex $u\in U$ with
probability $\rho/\hat{k}$ into a subset $U'$, then computing
$\reach{1}(U')$.  If $C(U_0,v) \gg \hat{k}/\rho $, then it is reasonably
likely that $v \in \reach{1}(U')$.  Conversely, if $C(U_0,v) \ll
\hat{k}/\rho$, then it is likely that $v \not\in \reach{1}(U')$.
Repeating this process $\Theta(\log n)$ times and applying a Chernoff
bound allows us to correctly partition the vertices, with high
probability.

If $H$ is nonempty, then select any $y$ in $H$ and run STSP to compute
$U = \set{u \in U_0 | \dist{1}(u,y) < 0}$.  Finally, verify that
$\card{U} = \Omega(\hat{k}/\rho)$ just in case the estimation
procedure failed. 

If instead $H$ is empty, then $L = U_0$, and all vertices $v\in U_0$
should have $C(U_0,v) = O(\hat{k}/\rho)$.  Then it is straightforward
to construct a large random independent set. Select a
uniformly random subset $I' \subseteq U_0$ with
$\card{I'} = \Theta(\rho)$.  Then, set $I = I' - \reach{1}(I')$, where
``$-$'' here denotes set subtraction, which ensures that the set~$I$
is independent.  For each vertex in $v \in I'$, there is only a
constant probability that there is another vertex $u \in I', u\neq v$
such that $\dist{1}(u,v) < 0$.  Thus, as long as there no
negative-weight cycles in the graph, there is at least a constant
probability that $\card{I} = \Omega(\rho)$.  Repeating $\Theta(\log
n)$ times gives high probability of successfully finding an
independent set.

\subsection{The full algorithm}\seclabel{combining}

Assuming all of the lemmas stated in this section, we are almost ready
to prove \thmref{main-elim}.  The only remaining pieces are
eliminating an independent set and determining the appropriate value
for~$r$.

\paragraph{Eliminating an independent set.}  Let $I$ be a $1$-hop
independent set of negative vertices in the graph
$G = (V, E^+\cup E^-,w)$.  Then eliminating these edges is
straightforward.  Consider the subgraph $G^{\id{out}(I)}$.  Then
simply use the price function
$\phi(v) = \dist[G^{\id{out}(I)}]{1}(V,v)$, which can be computed by
running 1-hop BFD in $O(m\log n)$ time.  It is not too hard to see
that this price function accomplishes the task.

\begin{lemma}\lemlabel{ind}
  Let $G = (V,E^+\cup \id{out}(I), w)$ be a subgraph of the input
  graph, where $I$ is a 1-hop independent set of the
  negative vertices.  Suppose that $w(e) \geq 0$ for all $e \in E^+$.
  Then the price function given by $\phi(v) = \dist[G]{1}(V,v)$ is a
  valid price function that eliminates all negative-weight edges from~$G$.
\end{lemma}
\begin{proof}
  The claim is that $\dist{1}(V,v)=\dist{}(V,v)$.  If this claim is
  true, then $\phi(v)$ matches Johnson's price
  function~\cite{Johnson77}, and hence it eliminates all
  negative-weight edges.  Suppose for the sake of contradiction that
  there exists some $u$-to-$v$ path $p$ with $w(p) < \dist{1}(V,v)$.
  Split $p$ into subpaths at each vertex in $I$, giving rise to a
  sequence of nonempty subpaths $p_0,p_1,\ldots,p_\ell$.  Each subpath
  $p_i$ for $0 < i < \ell$ starts and ends at negative vertices and
  includes exactly one negative-weight edge.  Thus, because $I$ is a
  1-hop independent set, $w(p_i) \geq 0$ for $0<i<\ell$. If $u\in I$,
  then the subpath $p_0$ follows the same logic; if $u\not\in I$, then
  $p_0$ contains no negative edges. Either way, $w(p_0)\geq 0$.  We
  therefore have $w(p) = \sum_{i=0}^\ell w(p_i) \geq w(p_\ell)$, which
  contradicts the assumption that $p$ is a shorter path to $v$ than
  the 1-hop path $w(p_\ell)$. 
\end{proof}

\paragraph{Choosing $r$ to minimize runtime.}  Fixing
$\beta-1 = \tau = r$, there are two components that dominate the
running time of the algorithm: the betweenness reduction, with a
running time of $\tilde{O}(r^2m)$ (\lemref{betweenness-reduction}),
and eliminating the $r$-remote subset using hop reduction, with a
running time of $\tilde{O}((k^{1/3}/r)m + rm)$
(\lemref{hop-reduction}).  The total running time is thus
$\tilde{O}(m\cdot (r^2 + k^{1/3}/r))$, which is minimized by setting
$r=\Theta(k^{1/9})$, yielding $\tilde{O}(k^{2/9}m)$, as per
\thmref{main-elim}.

\paragraph{Proof of \thmref{main-elim}.}
Consider the steps of \algref{full}.  The first step is betweenness
reduction.  By \lemref{betweenness-reduction}, this step always either
correctly determines the graph contains a negative-weight cycle, or it
finds a valid price function $\phi_1$.  In the former case, the
algorithm terminates.  Otherwise, $\phi_1$ is valid so
$w_{\phi_1}(e) \geq 0$ for all $e \in E^+$ (which is a precondition of
the subsequent steps).  This step always takes $\tilde{O}(k^{2/9}m)$
time (\lemref{betweenness-reduction}).

The next step is to find a negative sandwich or 1-hop independent set.
By \corref{findsandwich}, there is an algorithm that always correctly
determine that the graph contains a negative-weight cycle, returns a
size $\Omega(k^{1/3})$ negative sandwich, or returns a
size-$\Omega(k^{1/3})$ 1-hop independent set.  This algorithm runs in
$\tilde{O}(m)$ time, with high probability.  Again, in the case of a
cycle, the algorithm terminates with a correct output.  In the case of
a 1-hop independent set, \lemref{ind} gives a way of finding a valid
price function that eliminates the independent set.  Thus, the
algorithm correctly eliminates $\Omega(k^{1/3})$ negative edges.  In
the last case, the algorithm continues to the next step.

The third step is to reweight the graph $G_\phi$ again to attempt to
establish remoteness.  From \lemref{sandwich}, the reweighting step
here always produces a valid price function~$\phi_2$.  Moreover,
\corref{sandwich} states that if the betweenness reduction was
successful in reducing the $\beta$-betweenness of all pairs of
vertices, then $U$ is $r$-remote in
$(G_{\phi_1})_{\phi_2} = G_{\phi_1+\phi_2}$.  By
\lemref{betweenness-reduction}, the betweenness reduction succeeds
with high probability, and thus this step also succeeds with high
probability.  (Otherwise, the entire algorithm restarts.)  When
proceeding past this point, $U$ is always an $r$-remote subset with
size $\Omega(k^{1/3})$ carrying over from the previous step, and
$\phi_1+\phi_2$ is always a valid price function.  This step takes
$\tilde{O}(k^{1/9}m)$ time from \corref{sandwich}.

The final step is to apply the hop-reduction technique on the graph
$G_{\phi_1+\phi_2}^{\id{out}(U)}$, where $\card{U} =
\Theta(k^{1/3})$
Because $U$ is $r$-remote, \lemref{hop-reduction} applies, indicating
that the $\Theta(k^{1/3})$ negative edges can be eliminated
deterministically in
$\tilde{O}(k^{2/9}m)$ time.   

Summing the running time of all steps gives $\tilde{O}(k^{2/9}m)$. \qed

%% file: between.tex
\section{Betweenness Reduction}
\seclabel{between}

This section expands on the problem of betweenness reduction,
introduced in \secref{over-btw}, with the goal of proving
\lemref{betweenness-reduction}.  Throughout this section, let
$G=(V,E^+\cup E^-, w)$ denote the input graph and let $n=\card{V}$ and
$m=\card{E^+ \cup E^-}$.  The variables $\beta$ and $\tau$ denote the
parameters for betweenness reduction, with $\beta \geq 1$ and
$1\leq \tau \leq \card{V}$.  Recall that the goal is to find a price
function $\phi$ such that for all vertices $u,v$, we have
$\bw[\phi]{\beta}(u,v) \leq n/\tau$.  The algorithm is parameterized by a
constant $c\geq 3$ used to adjust the probability of success.

\begin{algorithm}[t]
\caption{Algorithm for betweenness reduction}\alglabel{btw}
\SetKwInOut{Input}{input}
\Input{A graph $G=(V,E^+\cup E^-,w)$ with $w(e) \geq 0$ for $e\in
  E^+$}
\Input{Parameters $\tau$ and $\beta$ and constant $c>1$, with $\beta \geq 1$ and $1 \leq \tau \leq
  \card{V}$ }
\ \\
\nl    let $n = \card{V}$\\
\nl let $T\subseteq V$ be a uniformly random subset of $c\tau\ceil{\ln n}$
vertices\\
\nl \lilabel{btw-start}\ForEach{$x \in T$}{
  \nl run $\beta$-hop SSSP and STSP, computing $\beta$-hop distances
  from and to $x$, respectively}
\nl construct a new graph $H = (V, E_H, w_H)$ as follows:\\
$E_H = (T \times V) \cup (V\times T)$ \\
$w_H(u,v) = \dist[G]{\beta}(u,v)$ using precomputed distances to/from
vertices in $T$\\
\nl let $\ell = 2\card{T}$ (which equals $2c\tau\ceil{\ln n}$)\\
\nl compute super-source distances $d(v) =
\dist[H]{\ell}(V,v)$ and $d'(v) = \dist[H]{\ell+1}(V,v)$ for all $v
\in V$\\
\nl \lIf{$\exists v$ such that $d'(v) < d(v)$}{terminate algorithm and
  report ``cycle''}
\nl \lElse{\KwRet{price function $\phi = d$}}
\end{algorithm}

\algref{btw} presents the algorithm for betweenness reduction.  The
algorithm begins by sampling a subset $T$ of vertices with
$\card{T}=c\tau\ceil{\ln n}$ vertices.  The remainder of the algorithm is
devoted to reweighting the graph so that all $\beta$-hop distances to
or from vertices in $T$ become nonnegative.

There are many straightforward ways to accomplish the goal of
nonnegative $\beta$-hop distances to/from $T$; \algref{btw} is just
one concrete example.  \algref{btw} proceeds by computing all
$\beta$-hop distances from each vertex in $T$ and all $\beta$-hop
distances to each vertex in $T$, using SSSP and STSP, respectively.
Then, an auxiliary graph $H$ is constructed.  The graph $H$ contains
all edges of the form $(x,v)$ and $(v,x)$ where $x\in T$ and $v\in V$.
Thus, all edges in $H$ are, by construction, incident on a vertex in
$T$.  The weights of these edges are the corresponding $\beta$-hop
distances in~$G$ that have already been computed.  The final step of
the algorithm is to apply Johnson's strategy~\cite{Johnson77} to~$H$.
That is, compute distances to each vertex using super-source shortest
paths.  Because all edges are incident on a vertex in $T$, the
computation stops at $2\card{T} + 1$ hops, at which point either the
algorithm has discovered a negative-weight cycle, or the
$2\card{T}$-hop distances are the actual shortest path distances
in~$H$.  Finally, these distances are returned as a price function
for~$G$.

There are two main aspects of correctness to prove.  (1) The algorithm
finds a price function $\phi$ such that all $\beta$-hop distances
to/from $x\in T$ in $G_\phi$ are nonnegative.  The idea here is that
from Johnson's strategy~\cite{Johnson77}, the shortest-path distances
in $H$ constitute a valid price function $\phi$ that eliminates all
negative edges in~$H$.  These edges in $H$ correspond to $\beta$-hop
paths in $G$ to/from vertices in~$T$.  Thus, applying $\phi$ to $G$
ensures that these $\beta$-hop paths have nonnegative weight.  (2) The
algorithm reduces the betweenness of all pairs to at most
$n/\tau$, as discussed in \secref{over-btw}.  The claims, along with
running time, are proved next.

\begin{lemma}\lemlabel{btw-reweighting}
  Consider an execution of \algref{btw} on input graph $G(V,E^+\cup E^-,w)$
  starting from \liref{btw-start} with any arbitrary subset
  $T\subseteq V$.  (That is, this claim does not rely on any
  randomness of the sample.) Then we have the following:
  \begin{closeitemize}
  \item If the algorithm reports a negative-weight cycle, then $G$
    contains a negative-weight cycle.
  \item Otherwise, the algorithm returns a price function $\phi$ such
    that for all $v \in V$ and $x\in T$:
    $\dist[G_\phi]{\beta}(x,v) \geq 0$ and
    $\dist[G_\phi]{\beta}(v,x) \geq 0$. Moreover, if the initial
    weight satisfies $w(e) \geq 0$ for
    all $e \in E^+$, then the price function is valid.
  \end{closeitemize}       
\end{lemma}
\begin{proof}
  Let us start with the following observation: all simple paths in $H$
  have size at most $2\card{T}$, which follows from the fact that all
  edges in $H$ are incident on vertices in~$T$.  (If the path is
  larger, some vertex in $T$ has at least 2 incoming or outgoing
  edges, and hence the path is not simple.)  Simple paths therefore
  also have at most $2\card{T}$ negative-weight edges.  Thus, $H$ has
  a negative-weight cycle if and only if there exists a vertex $v$
  such that $\dist[H]{\ell+1}(V,v) < \dist[H]{\ell}(V,v)$, where $\ell
  = 2\card{T}$. We have
  thus established that a cycle is reported if and only if $H$ has a
  negative-weight cycle.  Moreover, if no cycle is reported, then
  $d(v)$ is the actual super-source distance in~$H$, so the
  standard (not hop-limited) triangle inequality applies to~$d$.

  Next, suppose that $H$ has a negative-weight cycle~$C$. Then it is easy
  to see that $G$ does as well: replace each edge in $C$ with the
  corresponding $h$-hop path in $G$, which by construction has the same
  weight. Therefore, when the algorithm reports a negative-weight
  cycle, that result is correct.

  For the remainder, suppose that there is no negative-weight cycle
  in~$H$, so a price function is returned. Here we prove the claim
  that the distance to/from each sample is nonnegative.  By the standard
  triangle inequality, for all $x\in T$ and $v \in V$ (and hence
  $(x,v)\in E_H$), we have
  $d(v) \leq d(x) + w_H(x,v) = d(x) + \dist[G]{\beta}(x,v)$, or
  $\dist[G]{\beta}(x,v) + d(x) - d(v) \geq 0$.  Setting $\phi = d$ and
  using \lemref{phidist}, we thus have
  $\dist[G_\phi]{\beta}(x,v) = \dist[G]{\beta}(x,v) + \phi(x) -
  \phi(v) = \dist[G]{\beta}(x,v) + d(x) - d(v) \geq 0$.  Similarly, by the
  symmetric argument now considering the edge $(v,x) \in E_H$, we have
  $d(x) \leq d(v) + w_H(v,x) = d(v) + \dist[G]{\beta}(v,x)$, or
  $\dist[G]{\beta}(v,x) + d(v) - d(x) \geq 0$.  Thus
  $\dist[G_\phi]{\beta}(v,x) = \dist[G]{\beta}(v,x) + \phi(v) -
  \phi(x) \geq 0$.

  Finally, let us address the validity of the price function~$\phi=d$.
  We shall again prove this using the triangle inequality.  The only
  issue here is that $H$ does not include all edges in $E^+$, so we
  cannot directly apply the triangle inequality on computed distances
  to these edges.  Start by noting that $d(v) \leq 0$ for all $v\in V$
  from the empty path.  Now consider any edge $(u,v) \in E^+$ and
  suppose $w(u,v) \geq 0$.  If $d(u) = 0$ then trivially
  $w_\phi(u,v) = w(u,v) - d(v) \geq w(u,v) \geq 0$.  Suppose instead
  that $d(u) < 0$.  Then a shortest path to $u$ in $H$ is nonempty and
  must end with a last edge $(x,u) \in E_H$ for some $x \in T$; that
  is, $d(u) = d(x) + w_H(x,u)$.  By \lemref{triangle} on~$G$, for
  $(u,v)\in E^+$ we have
  $\dist[G]{\beta}(x,v) \leq \dist[G]{\beta}(x,u) + w(u,v)$, and hence
  $w_H(x,v) \leq w_H(x,u) + w(u,v)$.  Thus, using the triangle
  inequality in $H$ we have
  $d(v) \leq d(x) + w_H(x,v) \leq d(x) + w_H(x,u) + w(u,v) = d(u) +
  w(u,v)$, or $w_\phi(u,v) = w(u,v) + d(u) - d(v) \geq 0$.
\end{proof}

\begin{lemma}\lemlabel{btw-time}
  Suppose that the input graph $G=(V,E^+\cup E^-,w)$ satisfies
  $w(e) \geq 0$ for all $e \in E^+$.  Let $m=\card{E^+\cup E^-}$ and
  $n=\card{V}$.  Then there is a realization of \algref{btw} that runs
  in $O(\beta\tau \log n (m+n\log n) + \tau^2 n\log^2n)$ time.
\end{lemma}
\begin{proof} The two dominant costs of the algorithm are the
  hop-limited SSSP computations.
  
  Since $w(e) \geq 0$ for all $e\in E^+$, we can apply \lemref{BFD} to
  compute $\beta$-hop distances in $G$, giving a running time of
  $O(\beta(m+n\log n))$ per such computation.  The cost of computing
  $\beta$-hop distances to and from vertices in $T$ is thus
  $\card{T} \cdot O(\beta(m+n\log n)) = O((\tau \log n)\cdot\beta (m +n\log
  n))$
  
  Next, consider the single super-source computation in $H$.  By
  construction,
  $\card{E_H} \leq 2\card{T \times V} = \Theta(\tau n\log n)$.
  Applying BFD to the $\Theta(\tau \log n)$-hop SSSP in $H$ gives a
  running time of
  $O((\tau \log n) \cdot (\tau n \log n + n\log n))$.\footnote{In
    fact, this bound can be improved to remove the $n\log n$ term by
    observing (as in the start of the proof of \lemref{btw-reweighting}) that we do
    not actually need BFD here---$\Theta(\tau \log n)$ rounds of
    Bellman-Ford suffice.  But given that we have not established
    notation for ``$h$ rounds of Bellman-Ford,'' the weaker bound is
    used here.}
\end{proof}

When $\beta - 1 = \tau = \Theta(r)$, this bound simplifies to $O((r^2\log n)
\cdot (m+n\log n)) =\tilde{O}(r^2 m)$.

\begin{lemma}\lemlabel{btw-highprob}
  Consider an execution of \algref{btw} on input graph
  $G=(V,E^+\cup E^-,w)$ and let $n=\card{V}$.  Then with probability at least
  $1-1/n^{c-2}$, the algorithm either
  \begin{closeitemize}
  \item correctly reports a negative-weight cycle, or
  \item returns a
  price function $\phi$ such that for all $u,v \in V$,
  $\bw[\phi]{\beta} \leq n/\tau$.  
\end{closeitemize}
\end{lemma}
\begin{proof}
  Consider a particular pair $u,v \in V$.  The proof focuses on
  showing that the claim holds with high probability for this
  pair. Then taking a union bound across all $n^2$ pairs proves the
  lemma.  All distances in this proof are distance in $G$ or
  $G_\phi$, so the subscript $G$ is omitted.

  Number all the vertices in $V$ as $x_1, x_2, \ldots, x_n$ such that
  $\thru{\beta}(u,x_1,v) \leq \thru{\beta}(u,x_2,v) \leq \cdots \leq
  \thru{\beta}(u,x_n,v)$.  Now let $y = x_j$ be the sampled vertex
  with lowest index/rank in the numbering.  If the algorithm reports a
  cycle, then by \lemref{btw-reweighting} this reporting is correct.
  For the remainder, suppose instead that the algorithm returns a price
  function $\phi$.

  By \lemref{btw-reweighting}, $\phi$ is such that
  $\dist[\phi]{\beta}(u,y) \geq 0$ and
  $\dist[\phi]{\beta}(y,v) \geq 0$ and hence
  $\thru[\phi]{\beta}(u,y,v) \geq 0$.  From \lemref{phidist}, for all
  $x \in V$, we have
  $\thru[\phi]{\beta}(u,x,v) = \dist{\beta}(u,x) + \phi(u) -\phi(x) +
  \dist{\beta}(x,v) + \phi(x) - \phi(v) = \thru{\beta}(u,x,v) +
  \phi(u) - \phi(v)$.  Moreover, for all $x_i$ with $i \geq j$, we
  have $\thru{\beta}(u,x_i,v) \geq \thru{\beta}(u,y,v)$, and hence
  $\thru[\phi]{\beta}(u,x_i,v) = \thru{\beta}(u,x_i,v) + \phi(u) -
  \phi(v) \geq \thru{\beta}(u,y,v) + \phi(u) - \phi(v) =
  \thru[\phi]{\beta}(u,y,v) \geq 0$.  Thus, $\bw[\phi]{\beta}(u,v)
  \leq j-1$, where $x_j$ is the lowest-rank sampled vertex.  As long
  as $j \leq \ceil{n/\tau}$, we have $j-1 < n/\tau$ and hence
  $\bw[\phi]{\beta}(u,v) < n/\tau$.
  
  A failure event (the algorithm neither reports a cycle nor hits the
  betweenness guarantee) can thus only occur if $j>\ceil{n/\tau}$.
  The last step of the proof is to bound this probability.
  For $j$ to be this large, each sample must be drawn from the
  $b=n-\ceil{n/\tau}$ other vertices.  If $b < \card{T}$, then there
  is never a failure. Otherwise, the failure probability is given by 
  $\left(\frac{b}{n}\right)\left(\frac{b-1}{n-1}\right)\left(\frac{b-2}{n-2}\right)\cdots\left(\frac{b-\card{T}+1}{n-\card{T}+1}\right)
  \leq \left(\frac{b}{n}\right)^{\card{T}} =
  \left(1-\frac{\ceil{n/\tau}}{n}\right)^{\card{T}} \leq
  \left(1-\frac{1}{\tau}\right)^{\card{T}} \leq (1-1/\tau)^{c\tau \ln n}
  \leq (1/n)^c$.  
\end{proof}

\paragraph{Proof of \lemref{betweenness-reduction}.} By assumption in
the lemma statement,
$w(e) \geq 0$ for all $e \in E^+$. 
Thus
\lemref{btw-time} can be applied, and the algorithm always meets the
promised running time.  Moreover, by \lemref{btw-reweighting}, the
algorithm always either correctly reports a cycle or returns a valid
price function.  Finally,
\lemref{btw-highprob} states that algorithm is successful with high
probability, in which case it reports a cycle or a price function with
the desired $\beta$-betweenness guarantee. \qed

%% file: findsandwich.tex
\section{Finding a Negative Sandwich or Independent Set}\seclabel{findsandwich}

This section expands on the problem of finding a negative sandwich or
independent set, as introduced in \secref{over-findsandwich}.  The
bulk of this section is devoted to proving \lemref{sandwich-crust}.
Recall that the input comprises the graph $G=(V,E^+\cup E^-,w)$, a
subset $U_0$ of negative vertices with $\hat{k} = \card{U_0}$, and
integer parameter $\rho$ with $1 \leq \rho \leq \hat{k}$.  

As outlined in \secref{over-findsandwich}, the first task of
\lemref{sandwich-crust} is to partition the negative vertices in $U_0$
into a heavy and light set.

\begin{algorithm}[t]
\caption{Algorithm to partition into heavy and light sets}\alglabel{partition}
\SetKwInOut{Input}{input}
\SetKwInOut{Output}{output}
\SetKwFunction{partition}{HL-Partition}
\SetKwProg{Fn}{}{}{end}

\Input{A graph $G=(V,E^+\cup E^-,w)$ with $w(e) \geq 0$ for $e\in
  E^+$}
\Input{Subset $U_0$ of negative vertices and integer $\rho$ with
  $1\leq \rho \leq \card{U_0}$}
\Output{A partition $\ang{H,L=U_0-H}$ of $U_0$}
\  \\
\Fn{\partition{$G=(V,E^+\cup E^-,w)$, $U_0$, $\rho$}}{
\nl let $\hat{k} = \card{U_0}$\\
\nl \lForEach{$v \in V$}{$\id{count}(v)
  \gets 0$}  
\nl \For{$c \ceil{\ln n}$ times}{
\nl generate set $U'$ by sampling each vertex in $U_0$ with probability
$\rho/\hat{k}$\\
\nl compute $R = \reach[G]{1}(U')$\\
\nl \lForEach{$v \in R$}{$\id{count}(v) \gets
  \id{count}(v) + 1$}
} 
\nl $H \gets \set{u\in U_0 | \id{count}(u) \geq (c/2)\ceil{\ln n}}$\\
\nl $L \gets U_0 - H$\\
\nl \Return{$\ang{H,L}$}
} 
\end{algorithm}

The partitioning algorithm is given by \algref{partition}.  The
algorithm is parameterized by a constant $c\geq 6$ that controls the
probability of failure.  The algorithm is straightforward.  Sample
each vertex in $U_0$ independently with probability $\rho/\hat{k}$ to
get a random subset $U'$. For each vertex in the $1$-hop reach of
$U_0$, increment a counter.  Repeat this process $c \ceil{\ln n}$ times.
Finally, the set $H$ is the set of vertices in $U_0$ with counts at
least $(c/2)\ceil{\ln n}$.

To prove the algorithm works, recall that
$C(U_0,v) = \card{\set{u \in U_0 | \dist{1}(u,v) < 0}}$.  Define a
vertex as \defn{heavy} if $C(U_0,v) \geq 2\hat{k}/\rho$ and
\defn{light} if $C(U_0,v) \leq (1/8)\hat{k}/\rho$.  (Some vertices are
neither heavy nor light.)

\begin{lemma}\lemlabel{partition}
  Consider an execution of \algref{partition} with input $G$, $U_0$,
  $\rho$.  Then with probability at least $1-1/n^{c/3-1}$, the partition
  is such that all heavy vertices in $U$ are in $H$ and all light vertices are in $L$.  Equivalently,
  with high probability:
  $\forall v \in H, C(U_0,v) > (1/8)\hat{k}/\rho$ and
  $\forall v \in L, C(U_0,v) < 2\hat{k}/\rho$.
\end{lemma}
\begin{proof}
  Consider a heavy vertex~$v \in U_0$.  Let $X_i$ be the indicator
  that $\id{count}(v)$ increases in the $i$th iteration of the loop,
  and let $X = \id{count}(v) = \sum_{i=1}^{c\ceil{\ln n}} X_i$.  In each
  iteration of the loop, $\Pr(X_i = 0)$ is the probability that none
  of the vertices that can reach $v$ are sampled, which is
  $\Pr(X_i = 0) \leq (1-\rho/\hat{k})^{C(U_0,v)} \leq
  (1-\rho/\hat{k})^{2\hat{k}/\rho} \leq 1/e^2$. Let
  $p = E[X_i]$. Then $p = \Pr(X_i = 1) \geq (1-1/e^2) > 6/7$.  Because
  the $X_i$'s are independent identically distributed indicators, we can apply a
  Chernoff-Hoeffding bound to get $\Pr(X \leq (1/2)c\ceil{\ln n})$.  In
  particular, set $\epsilon = p - 1/2$ or $1/2=p-\epsilon$.  Then we
  have
  $\Pr(X \leq (1/2)c\ceil{\ln n}) = \Pr(X \leq (p-\epsilon) c\ceil{\ln n}) \leq
  \left(\left(\frac{p}{1/2}\right)^{1/2}
    \left(\frac{1-p}{1/2}\right)^{1/2}\right)^{c\ceil{\ln n}} \leq
  (1/e)^{(1/3)c\ln n} = 1/n^{c/3}$ when $p\geq 6/7$. 

  Consider instead a light vertex $v$.  Again, let $X_i$ be the
  indicator that $\id{count}(v)$ increases in the $i$th iteration and
  $X = \sum_{i=1}^{c\ceil{\ln n}}X_i$.  Now we have
  $E[X_i] \leq (\rho/\hat{k})\cdot C(U_0,v) \leq 1/8$ by the union
  bound.  Let $p=E[X_i] \leq 1/8$.  Again, the $X_i$'s are i.i.d.\
  indicators, so the Chernoff-Hoeffding bound applies.  In particular,
  set $\epsilon = 1/2-p$ or $p+\epsilon=1/2$.  Then we have
  $\Pr(X \geq (1/2)c\ceil{\ln n}) = \Pr(X \geq (p+\epsilon)c\ceil{\ln n}) \leq
  \left(\left(\frac{p}{1/2}\right)^{1/2}
    \left(\frac{1-p}{1/2}\right)^{1/2}\right)^{c\ceil{\ln n}} \leq
  (1/e)^{(1/3)c\ln n} = 1/n^{c/3}$ when $p \leq 1/8$.  

  Taking the union bound across all vertices in $U_0$, the probability
  that any heavy or light vertex is misclassified is at most
  $1/n^{c/3-1}$.  This bound is only meaningful if $c$ is strictly
  larger than~$3$.
\end{proof}

\begin{algorithm}[t]
\caption{Algorithm to find a random $1$-hop independent set}\alglabel{indset}
\SetKwInOut{Input}{input}
\SetKwInOut{Output}{output}
\SetKwFunction{ind}{RandIS}
\SetKwProg{Fn}{}{}{end}

\Input{A graph $G=(V,E^+\cup E^-,w)$ with $w(e) \geq 0$ for
  $e\in E^+$} \Input{Subset $U_0$ of negative vertices and integer
  $\rho$ with $1\leq \rho \leq \card{U_0}$} \Output{A $1$-hop
  independent set $I \subseteq U_0$}
\  \\
\Fn{\ind{$G=(V,E^+\cup E^-,w)$, $U_0$, $\rho$}}{
  \nl let $I'$ be a uniformly random size-$\ceil{\rho/4}$ subset of $U_0$\\
  \nl solve the super-source problem to compute
  $d(v) = \dist[G]{1}(I',v)$ and also a corresponding starting vertex
  $s(v) \in I'$ such that $d(v) =
  \dist[G]{1}(s,v)$\\
  \nl \ForEach{$u \in I'$}{\nl \lIf{$d(u) < 0$ and $s(u) = u$}{terminate
      algorithm and report ``cycle''}}
  \nl $R \gets \set{v | d(v) < 0}$\\
  \nl $I \gets I' - R$\\
  \nl \Return{$I$}
} 
\end{algorithm}

Now let us turn to the task of finding an independent set in the event
that the returned partition has $H = \emptyset$.  The algorithm is
given by \algref{indset}.  The algorithm is simple: sample a uniformly
random
size-$\ceil{\rho/4}$ subset $I'$ of $U_0$, and then remove from $I'$
any vertices than can be reached by negative-weight $1$-hop paths from
any other vertex in $I'$.  It is easy to see that this set is now
a $1$-hop independent set.

There is one other issue: if there are
negative-weight $1$-hop cycles in~$G$, then we cannot bound the
likelihood that the independent set is large.  Thus, the algorithm
also checks whether any of the shortest paths computed by the
black-box subroutine correspond to negative-weight cycles.  In
particular, recall that for the super-source version of the problem,
\lemref{BFD} states that BFD (and indeed any relaxation-based SSSP algorithms)
can be augmented to return some vertex $s(v) \in I'$ such that
$\dist[G]{1}(I',v) = \dist[G]{1}(s(v),v)$.  If $s(v) = v$ and the
distance to~$v$ is negative, then a negative-weight cycle is reported.  Once
a cycle is reported, the entire algorithm terminates.

\begin{lemma}\lemlabel{indset}
  Consider an execution of \algref{indset} with input $G$, $U_0$,
  $\rho$.  The algorithm always correctly reports a negative-weight
  cycle (i.e., only if $G$ has a negative-weight cycle) or returns a 1-hop
  independent set $I\subseteq U_0$.

  Suppose that there are no heavy vertices in~$U_0$.  Then the
  probability that the algorithm returns an independent set with
  $\card{I} < \rho/16$ is at most $5/6$.  Conversely, with probability
  at least $1/6$, the algorithm either correctly reports a cycle or
  returns an independent set with~$\card{I} \geq \rho/16$. 
\end{lemma}
\begin{proof}
  The algorithm only reports a cycle if there is a vertex $v$ such
  that $\dist[G]{1}(v,v) < 0$, and thus there is a negative cycle.
  Now suppose the algorithm returns a set $I$, and assume for the sake
  of contradiction that $I$ is not a $1$-hop independent set.  Then
  there exists a pair $u,v \in I \subseteq I'$ with $\dist[G]{1}(u,v) < 0$.  But in this
  case, $v$ would be removed from $I$, which contradicts the assumption.
  
  We now turn to the claim about~$\card{I}$.  Note that if
  $\rho \leq 4$, then $\card{I'} = 1$.  Thus, there is either a
  negative-weight cycle discovered, or
  $\card{I} = \card{I'} = \ceil{\rho/4} > \rho/16$; either way, the claim
  holds.  For the remainder, assume $\rho > 4$ and hence
  $\hat{k} > 4$.

  As per the lemma statement, assume there are no heavy vertices
  in~$U_0$.  We say that a knockout event occurs for $v$ if (i)
  $v \in I'$ and (ii) $\exists u \in I'$ with $u\neq v$ such that
  $\dist[G]{1}(u,v) < 0$.  Let $X_v$ be the indicator for a knockout
  event for $v$.  We can bound $\Pr(X_v | v \in I')$ as follows.
  Consider any $u\neq v$, $u\in U_0$.  Then
  $\Pr(u \in I' | v \in I') = \frac{\card{I'}-1}{\card{U_0}-1} =
  \frac{\ceil{\rho/4}-1}{\hat{k}-1}$.  Taking a union bound over all
  $u\neq v\in U_0$ with $\dist[G]{1}(u,v) <0$, we have
  $\Pr(X_v | v\in I') \leq C(U_0,v) \frac{\ceil{\rho/4}-1}{\hat{k}-1}
  \leq \frac{2\hat{k}}{\rho} \frac{\rho/4}{\hat{k}-1} =
  \frac{\hat{k}}{2(\hat{k}-1)}$, or
  $E[X_v | v \in I'] \leq \frac{\hat{k}}{2(\hat{k}-1)}$.  For
  $\hat{k} \geq 5$, this can be simplified to $E[X_k] \leq 5/8$.  Let
  $X$ be the total number of knockout events. Then we have
  $E[X] \leq (5/8) \card{I'}$.  By Markov's inequality, we then have
  $\Pr(X \geq (3/4) \card{I'}) \leq 5/6$.

  Now, let us consider the ramifications of the good outcome: $X < (3/4)\card{I'}$
  knockout events.  If $\card{I} = \card{I'} - X$, then
  $\card{I} \geq (1/4) \card{I'} = (1/4) \ceil{\rho/4} \geq \rho/16$.
  If instead $\card{I} < \card{I'}-X$, then there must be some vertex
  $v \in I'$, $v \not\in I$ that is removed for a reason other than a
  knockout event.  That is to say, $v \in \reach[G]{1}(I')$ but
  $v\not\in \reach[G]{1}(I'-\set{v})$.  Thus, $\dist[G]{1}(v,v) < 0$
  and $s(v) = v$, and a cycle is reported.  We conclude that with
  probability at least $1/6$, the number of knockout events is small
  enough and hence either $\card{I}\geq \rho/16$ or a cycle is reported.
\end{proof}

With all the tools in place, we are ready to complete the algorithm
for \lemref{sandwich-crust}, which is described in~\algref{crust}.
This algorithm is parameterized by a constant $c'\geq 4$, which
controls the failure probability.  The process matches the outline in
\secref{over-findsandwich}.  First partition the negative vertices
$U_0$ into subsets $H$ and $L$, where $H$ should contain the heavy
vertices and $L$ should contain the light vertices, using
\algref{partition}.  If $H$ is nonempty, then choose any vertex $y$
and identify the set of negative vertices
$U = \set{u \in U_0 | \dist[G]{1}(u,y)<0}$.  This can be accomplished
by computing 1-hop STSP to $y$ using BFD.  As this is supposed to be a
Las Vegas algorithm, the next step is to verify that $U$ is large
enough.  If so, return $y$ and $U$.  If not (some vertex was
misclassified), restart the algorithm.  If instead $H$ is empty, then
the algorithm instead searches for a large independent set
$I \subseteq U_0$ by calling \algref{indset} a total of $c' \ceil{\lg n}$
times, stopping when either a cycle is reported or a large independent
set is found.  This step may also fail either because we are unlucky
or because some heavy vertices were misclassified in $L$.  Thus, after
$c' \ceil{\lg n}$ failed attempts, the algorithm is restarted.

\begin{algorithm}[t]
\caption{Algorithm of \lemref{sandwich-crust}: find a sandwich crust
  or independent set}\alglabel{crust}
\SetKwInOut{Input}{input}
\SetKwInOut{Output}{output}
\SetKwFunction{partition}{HL-Partition}
\SetKwFunction{ind}{RandIS}

\Input{A graph $G=(V,E^+\cup E^-,w)$ with $w(e) \geq 0$ for
  $e\in E^+$}
\Input{Subset $U_0$ of negative vertices and integer
  $\rho$ with $1\leq \rho \leq \card{U_0}$}
\Output{A $1$-hop
  independent set $I\subseteq U_0$ or a vertex $y$ and set $U\subseteq
  U_0$ such that $\dist[G]{1}(u,y) < 0$ for all $u\in U$.  A
  negative-weight cycle my instead be reported inside a call to \ind,
  which terminates the entire algorithm}
  \  \\
  \nl let $\hat{k} = \card{U_0}$\\
  \nl $\ang{H,L} \gets \partition(G,U_0,\rho)$\\
  \nl \If{$H \neq \emptyset$}{
  \nl choose arbitrary $y \in H$\\
  \nl run STSP with target $y$ to compute $U = \set{u \in U_0 |
    \dist[G]{1}(u,y) < 0}$\\
  \nl \lIf{$\card{U} < (1/8)\hat{k}/\rho$}{restart \algref{crust}}
  \nl \lElse{\Return{$y$ and $U$}}
      }
   \tcc{we now have $H=\emptyset$ and $L=U_0$}
  \nl \For{$c' \ceil{\lg n}$ attempts}{
  \nl $I \gets \ind(G,U_0,\rho)$\\   
  \nl \lIf{$\card{I} \geq \rho/16$}{\Return{$I$}}
}
\tcc{no large independent set found}
  \nl restart \algref{crust}
\end{algorithm}

\paragraph{Proof of \lemref{sandwich-crust}.} First, we consider the
return values. By \lemref{indset}, if \algref{indset} reports a cycle,
then that reporting is always correct.  Also by \lemref{indset}, the
set $I$ is always a $1$-hop independent set.  Thus, if \algref{crust}
returns $I$, then $I$ is a 1-hop independent set with
$\card{I} \geq \rho/16$.  Finally, by construction,
$U = \set{u\in U_0 | \dist[G]{1}(u,y) < 0}$, and the algorithm only
returns $U$ and $y$ if $\card{U} \geq (1/8)\hat{k}/\rho$. There are no
other places where \algref{crust} returns, so it always satisfies the
output criteria of this lemma.

We next consider the running time.  Because $w(e) \geq 0$ for all
$e \in E^+$, we can apply \lemref{BFD} to compute SSSP and STSP.
First, let us consider the running time of \partition
(\algref{partition}). Computing $\reach[G]{1}(U')$ amounts to
computing $1$-hop SSSP, which takes time $O(m\log n)$ from
\lemref{BFD}.  The random sampling and set construction can be
performed within this time complexity as well, so the time of
\partition is $O(m\log^2n)$ for the $\Theta(\log n)$ iterations.
There is a potential partition failure event: that some vertex is
misclassified in $L$ or $H$.  By \lemref{partition}, the probability
of such a failure is at most $1/n^{c/3-1}$.

Suppose there is no partition failure. Then
$C(U_0,v) > (1/8)\hat{k}/\rho$ for all $y \in H$.  Thus, if $H$ is not
empty, then the algorithm always returns a $y$ and $U$.  This step
entails running 1-hop STSP again, which is $O(m\log n)$ time from
\lemref{BFD}.

If instead there is no partition failure, but $H=\emptyset$, then the
algorithm proceeds to finding an independent set.  Each call to \ind
(\algref{indset}) entails computing 1-hop SSSP and scanning through
the vertices once, so $O(m\log n)$ time.  There are $c'\ceil{\lg n}$ such
calls, so the running time is again $O(m\log^2 n)$.  By
\lemref{indset}, which also assumes no partition failure, each call to
\ind leads to a probability of $5/6$ that \algref{crust} completes,
either finding a large-enough independent set or reporting a cycle and
terminating.  Thus, the probability that the algorithm \emph{does not}
complete by the end of the loop is at most $(5/6)^{c'\ceil{\lg n}} =
1/n^{c'\lg(6/5)} < 1/n^{c'/4}$.

To conclude, the \algref{crust} completes in $O(m\log^2n)$ time unless
there is a partition failure or there is an unlucky outcome with
independent sets, either of which may result in the algorithm
restarting.   Adding up the failure probabilities gives a failure
probability of at most $1/n^{c'/4}+1/n^{c/3-1}$.
Choosing, for example, $c=9$ and $c' = 8$ gives a failure
probability of at most $2/n^2$.\qed

%% file: usesandwich.tex
\section{Reweighting a Negative Sandwich}\seclabel{usesandwich}

This section provides a proof of \lemref{sandwich}.  Recall that the
lemma states that given input graph $G$ and negative sandwich
$(x,U,y)$, (1) the specific reweighting $\phi$ is valid, and (2) that
the only vertices in $\reach[\phi]{\beta-1}(U)$ after reweighting are
those vertices $v$ for which $\thru{\beta}(x,v,y) < 0$ before. 

\paragraph{Proof of \lemref{sandwich}.}
  Throughout the proof, we use $\dist{}$ for the distance in $G$, i.e.,
  with weight function $w$, and $\dist[\phi]{}$ for the distance in
  $G_\phi$, i.e., with weight function $w_\phi$.  The latter only
  occurs at one point in the proof of (2).
  
  To prove (1), consider any nonnegative edge $(u,v)\in E^+$.  We then
  have three cases.

  \noindent Case 1: $\phi(u) = 0$.  We always have $\phi(v) \leq 0$.
  So $w_\phi(u,v) = w(u,v) + \phi(u) - \phi(v) = w(u,v) + 0 - \phi(v)
  \geq w(u,v) \geq 0$.

  For the remaining two cases, observe first the following
  \begin{equation}
    \max(\dist{\beta}(x,v),-\dist{\beta}(v,y)) \geq \phi(v) \eqnlabel{phiv}
  \end{equation}
  \begin{equation}
    (\phi(u) \neq 0) \implies ((\phi(u) \geq \dist{\beta}(x,u))
    \wedge (\phi(u) \geq -\dist{\beta}(u,y))) \eqnlabel{phiu}
  \end{equation}
  
  \noindent Case 2: $\phi(u) \neq 0$ and $\dist{\beta}(x,v) \geq
  -\dist{\beta}(v,y)$.  
  By the triangle inequality
  (\lemref{triangle}), $\dist{\beta}(x,v) \leq
  \dist{\beta}(x,u) + w(u,v)$ or equivalently $\dist{\beta}(x,u)
  \geq \dist{\beta}(x,v)-w(u,v)$. Putting everything together
  \begin{align*}
    \phi(u) &\geq \dist{\beta}(x,u) &&\text{\eqnref{phiu}}\\
            &\geq \dist{\beta}(x,v)-w(u,v) && \text{triangle inequality}\\
            &\geq \phi(v) - w(u,v) && \text{\eqnref{phiv} with $\dist{\beta}(x,v) \geq
                                      -\dist{\beta}(v,y)$}\\
    \therefore w(u,v)+\phi(u)-\phi(v) &\geq 0 \ .
  \end{align*}

  \noindent Case 3: $\phi(u) \neq 0$ and $-\dist{\beta}(v,y) >
  \dist{\beta}(x,v)$.  By the triangle inequality
  (\lemref{triangle}), $\dist{\beta}(u,y) \leq w(u,v)
  + \dist{\beta}(v,y)$ or equivalently $-\dist{\beta}(u,y) \geq 
  -w(u,v) - \dist{\beta}(v,y)$.  Putting everything together
  \begin{align*}
    \phi(u) &\geq -\dist{\beta}(u,y) &&\text{\eqnref{phiu}}\\
            &\geq -w(u,v) - \dist{\beta}(v,y) &&\text{triangle inquality}\\
            &\geq -w(u,v) + \phi(v) &&\text{\eqnref{phiv} with
                                       $-\dist{\beta}(v,y) >
                                       \dist{\beta}(x,v)$}\\
    \therefore w(u,v) + \phi(u) - \phi(v) &\geq 0 \ .
  \end{align*}

  Finally, let us prove (2).   Consider any $u\in U$ and
  $v$ that is not $\beta$-between $x$ and $y$.  The goal is to argue
  that $\dist[\phi]{\beta-1}(u,v) \geq 0$. We proceed by breaking the
  proof into two smaller claims, namely (i) $\phi(u) = 0$ and (ii) $-\phi(v) >
  -\dist{\beta-1}(u,v)$.  Assuming these claims hold, we have
  $\dist[\phi]{\beta-1}(u,v) = \dist{\beta-1}(u,v) + \phi(u) -
  \phi(v) > \dist{\beta-1}(u,v) + 0 - \dist{\beta-1}(u,v) = 0$
  as desired.  

  Claim (i) follows from definition of a negative sandwich and~$\phi$.
  That is, $\dist{\beta}(u,y) \leq \dist{1}(u,y) < 0$.
  Therefore, $\max(\dist{\beta}(x,u),-\dist{\beta}(u,y)) \geq
  -\dist{\beta}(u,y) > 0$, and hence $\phi(u)=0$.  

  For claim (ii), start with the definition of $\beta$-betweenness.
  By assumption, $v$ is not $\beta$-between $x$ and $y$, so
  $\dist{\beta}(x,v) + \dist{\beta}(v,y) \geq 0$.  Therefore,
  $\phi(v) = \min(0,\dist{\beta}(x,v)) \leq \dist{\beta}(x,v)$.
  By the triangle inequality, $\phi(v) \leq \dist{\beta}(x,v) \leq
  \dist{1}(x,u) + \dist{\beta-1}(u,v)$.  Because of the negative
  sandwich $\dist{1}(x,u) < 0$, and hence $\phi(v) <
  \dist{\beta-1}(u,v)$, which completes the proof of (ii). \qed

%% file: hopreduction.tex
\section{Eliminating $r$-Remote Edges by Hop
  Reduction}\seclabel{hopreduction} 

This section proves \lemref{hop-reduction}, expanding on the
hop-reduction technique of \secref{over-hopreduction}.
\algref{hopreduction} provides pseudocode of the algorithm.  Recall
that the crux of the algorithm is building a new graph
$H=(V_H,E_H,w_H)$ so that $h$-hop paths in $G$ correspond to
$\leq \ceil{h/r}$-hop paths in~$H$.  This section proves that the
graph construction has this feature, and hence that SSSP distances can
be computed efficiently by instead computing distances in~$H$. 

Aside from the graph construction, the algorithm is
straightforward. \algref{hopreduction} begins by computing distances
$\delta_j(v) = \dist[G]{j}(V,v)$ in $G$ for $0\leq j \leq r$, which by
\lemref{BFD} corresponds to one $r$-limited SSSP computation.
These distances are used to construct $H$.  Next, the graph $H$ is constructed,
discussed more below.  Finally, the algorithm computes $\ceil{\hat{k}/r}$
and $(\ceil{\hat{k}/r}+1)$-hop distances in $H$.  If these are
different, the algorithm terminates by reporting a cycle.  If these
are the same, then the price function for $v \in V$ is given by
$\phi(v) = \dist[H]{\ceil{\hat{k}/r}}(V,v)$.

\begin{algorithm}[t]
\caption{Algorithm of \lemref{hop-reduction}: eliminate a remote
  subset by hop reduction}\alglabel{hopreduction}
\SetKwInOut{Input}{input}
\SetKwInOut{Output}{output}
\SetKwFunction{partition}{HL-Partition}
\SetKwFunction{ind}{RandIS}

\Input{Integer $r \geq 1$}
\Input{A graph $G=(V,E^+\cup N,w)$ with $w(e) \geq 0$ for
  $e\in E^+$}
\Output{A valid price function $\phi$ that eliminates all edges $N$.
  The algorithm may instead terminate by reporting a negative-weight
  cycle.}
  \ \\
  \nl let  $\hat{k} = \card{N}$\\
  \nl compute super-source distances $\delta_j(v) = \dist[G]{j}(V,v)$ for
  all vertices $v$ and all $j$,  $0\leq j \leq r$\\
  \nl $R \gets \set{v | \delta_r(v) < 0}$\\
  \nl construct a new graph $H=(V_H,E_H,w_H)$ as follows:\\
  $V_H = V \cup \set{v_j | v\in R, 1\leq j \leq r}$.\\
  \hspace{1cm} define $v_0 = v$ as an alias for $v$, for all $v\in V$\vspace{-1em}
  \begin{flalign*}\hspace{-1.5ex}
    E_H &= \set{(u_j,v_j) | (u,v) \in E^+, u,v\in R, 0\leq j \leq r}
    &&\cup \set{(u_j,v_{j+1}) | (u,v)\in N, u,v,\in R, 0\leq j <
      r}&\\
    &\cup \set{(u_j,v_0) | (u,v) \in E^+, u\in R, v\not\in R, 0\leq
      j \leq r}\hspace{-1em} &&\cup \set{(u_j,v_0) | (u,v)\in N, u\in R,
      v\not\in R, 0\leq j < r} &\\
    &\cup \set{(u_0,v_0)|(u,v)\in E^+, u\not\in R, v\in R} &&\cup
    \set{(u_0,v_1)|(u,v)\in N, u\not\in R, v\in R} &\\
    &\cup \set{(u_0,v_0) | (u,v) \in E^+, u,v\not\in R} &&\cup
    \set{(u_0,v_0)|(u,v)\in N, u,v \not\in R}&\\
    &\cup \set{(u_0,u_1), (u_1,u_2),\ldots,(u_{r-1},u_r),(u_r,u_0)|
      u\in R}\hspace{-1in} &
  \end{flalign*}\\\vspace{-1em}
  $w_H(u_i,v_j) = w(u,v) + \delta_i(u)
  - \delta_j(v)$ for $(u_i,v_j) \in E_H$\\
  \nl let $\kappa = \ceil{\hat{k}/r}$ \\
  \nl compute super-source distances $d(v) = \dist[H]{\kappa}(V,v)$ and
  $d'(v) =\dist[H]{\kappa+1}(V,v)$ for all $v \in V_H$\\
  \nl \lIf{$\exists v \in V_H$ such that $d'(v) < d(v)$}{terminate
    algorithm and report ``cycle''}
  \nl \lElse{\Return{price function $\phi: V\rightarrow \mathbb{R}$ with
      $\phi(v) = d(v)$ (i.e., $d$ restricted to subdomain $V$)}}
\end{algorithm}

\paragraph{Vertices $V_H$.}   For all of the following, let
$R = \set{v | \delta_r(v) < 0}$.  
All of the vertices in $V$ are also in $V_H$; define $v_0=v$, so when
referring to a vertex $v\in V$ in the context of the graph $H$, we may
use either $v_0$ or $v$.\footnote{The notation $v_0$ is generally used
  when considering distances or weights of edges in $H$, and the
  notation $v$ is generally used when relating the distances back to
  $G$.}  In addition, for each vertex $v\in R$, $V_H$ contains $r$
additional copies $v_1,v_2,\ldots, v_r$ of the vertex.  The subscript
$\ell$ in $v_\ell$ is called the \defn{layer} of the vertex. Layer 0
is the original vertices.

\paragraph{Edges $E_H$.} For the edges, there are several cases depending on whether the
endpoints are in $R$ or not, i.e., whether the endpoints occur in more
than one layer.  The cases are grouped in the pseudocode by endpoint
classifications across four rows and edge type ($E^+$ or $N$) across
the two columns. Let us consider the nonnegative edges $(u,v) \in E^+$
first.  The number of corresponding edges in $H$ is determined by whether $u\in R$, and the
target of the edges depends on whether $v\in R$.  If $u,v \in R$, then
there are $r+1$ copies of each endpoint, and there are $r+1$
corresponding copies $(u_0,v_0), (u_1,v_1), \cdots, (u_r,v_r)$ of the
edge included in $E_H$.  These edges are each within a single layer.
If $u \in R$ but $v\not\in R$, then there are still $r+1$ copies of
the edge, but they are all directed at $v_0$ in layer~0, i.e., the
edges have the form $(u_j,v_0)$ for $0\leq j \leq r$.  If instead
$u\not\in R$ then $u$ only occurs in layer~0, and hence there is only
a single copy of the edge $(u_0,v_0)$ in $E_H$.  Notice that for all
edges $(u,v) \in E^+$, the corresponding edges in $E_H$ have the form
$(u_j,v_j)$ or $(u_j,v_0)$---that is, these edges are never directed
toward a higher layer.  Moreover, for each $(u,v)\in E^+$, each
$u_i \in V_H$ has exactly one such outgoing edge. 

Now consider the negative edges $(u,v) \in N$.  Again, the number of
edges is dictated by whether $u\in R$, and the target depends on
whether $v\in R$.  If $u,v \in R$, then there are $r$ corresponding
copies $(u_0,v_1), (u_1,v_2), \ldots, (u_{r-1},v_r)$ of the edge in
$E_H$; here, each $(u_j,v_{j+1})$ progresses from layer $j$ to layer
$j+1$, which is the key difference in the construction for negative
edges and nonnegative edges.  If $u \in R$ but $v\not\in R$, then
there are still $r$ copies of the edge, but they all directed at
layer-0 vertex $v_0$, i.e., the edges have the form $(u_j,v_0)$ for
$0\leq j < r$.  If instead $u\not \in R$, then there is only one copy
of the edge in $E_H$: if $v \in R$, then the edge is $(u_0,v_1)$; if
$v\not\in R$, then the edge is $(u_0,v_0)$.  Unlike the nonnegative
case, these edges may be directed toward a higher layer, but it is
always at most one higher.  Specifically, for $(u,v)\in N$, the
corresponding edges all have the form $(u_j,v_{j+1})$ or $(u_j,v_0)$.
Moreover, for $(u,v) \in N$, each $u_i \in V_H$ with $i < r$ has
exactly one outgoing edge of the form $(u_i,v_j)$ (and moreover
$j\in\set{0,i+1}$).  The copy of $u_r$ in the $r$-th layer has no
corresponding outgoing edge as there is no layer $r+1$ to move to.  An
astute reader may notice that as described so far, a layer-$r$ copy of
a negative vertex (whose only outgoing edge in~$G$ is a negative edge)
would be a dead end in~$H$.  The self edges, discussed next, provide
an outgoing edge.

For $u\in R$,  $E_H$ also includes the self edges
$(u_j,u_{j+1})$ for $0\leq j < r$ and $(u_r,u_0)$.  These edges form a
cycle on copies of $u$, and the weights will be set so that this is a
0-weight cycle. These edges serve two purposes. First, the edges
$(u_r,u_0)$ provide routes from layer-$r$ to layer-$0$. Second, the
other edges in the cycle simplify the reasoning about
distances in~$H$. 

\paragraph{Weights $w_H$.} For each edge $(u_i,v_j) \in E_H$, the
weight is simply $w_H = w(u,v) + \delta_i(u) - \delta_j(v)$, where for
notational convenience we define $w(u,u) = 0$ for all $u\in V$.

\subsection{Analysis}

This section proves \lemref{hop-reduction}.  Let us begin by observing
that most edges in $H$ have nonnegative weight.  In particular, the
negative edges in $H$ are limited to the self edges $(u_r,u_0)$ from
layer $r$ to layer~0. The proof amounts to applying the triangle inequality
(\lemref{triangle}) for each of several cases.

\begin{lemma}\lemlabel{h-weights}
  Consider the input graph $G=(V,E^+\cup N,w)$ and auxiliary graph
  $H=(V_H,E_H,w_H)$ as constucted by \algref{hopreduction}.  The only
  edges $e\in E_H$ with $w_H(e) < 0$ are the edges $e\in
  \set{(u_r,u_0)}$
\end{lemma}
\begin{proof}
  Consider any edge $(u_i,v_j) \in E_H$.  Showing
  $w_H(u_i,v_j) \geq 0$ amounts to showing
  $w(u,v) + \delta_i(u) - \delta_j(v) \geq 0$, or
  $\delta_j(v) \leq \delta_i(u) + w(u,v)$, i.e., the triangle
  inequality but possibly with different numbers of hops.  It is easy
  to verify the claim by considering the cases separately: (1) edges
  $(u_i,u_{i+1})$, (2) edges $(u_i,v_0)$ for $v \not\in R$, (3) edges
  $(u_i,v_i)$ for $(u,v) \in E^+$, and (4) edges $(u_i,v_{i+1})$ for
  $(u,v) \in N$.

  Case 1.  Consider an edge $(u_i,u_{i+1}) \in E_H$.  Because $i$-hop
  paths are a subset of $(i+1)$-hop paths, $\delta_{i+1}(u) \leq
  \delta_i(u) = \delta_i(u) + 0 = \delta_i(u) + w(u,u)$.

  Case 2. Consider an edge $(u_i,v_0) \in E_H$ for $v\not\in R$.
  First, suppose $i<r$.  By \lemref{triangle},
  $\dist[G]{r}(V,v) \leq \dist[G]{i}(V,u) + w(u,v) = \delta_i(u) +
  w(u,v)$.  Because $v\not\in R$, $\dist[G]{r}(V,v) \geq 0$ (which
  means it equals 0), and hence $\delta_i(v) = 0$ for all $i$.  Thus,
  we have $\delta_0(v) = \delta_r(v) \leq \delta_i(u) + w(u,v)$. The
  case that $i=r$ only occurs for $(u,v) \in E^+$.  Then by
  \lemref{triangle}, again
  $\dist[G]{r}(V,v) \leq \dist[G]{r}(V,u) + w(u,v)$, and hence
  $\delta_0(v) = \delta_r(v) \leq \delta_r(u) + w(u,v)$.

  Case 3. Consider an edge $(u_i,v_i)$ for $(u,v) \in E^+$.  Then by
  \lemref{triangle}, $\dist[G]{i}(V,v) \leq \dist[G]{i}(V,u) + w(u,v)$
  or $\delta_i(v) \leq \delta_i(u) + w(u,v)$.

  Case 4. Consider an edge $(u_i,v_{i+1})$ for $(u,v) \in N$.  Then by
  \lemref{triangle}, $\dist[G]{i+1}(V,v) \leq \dist[G]{i}(V,u) +
  w(u,v)$ and $\delta_{i+1}(v) \leq \delta_i(u) + w(u,v)$.  
\end{proof}

The next lemmas show a correspondence between paths in $H$ and paths
in $G$.  The first, which is simpler, shows that paths between vertices
in $V$ in the graph $H$ correspond to paths in $G$, and moreover those
paths have the same weight.  The second roughly shows the converse,
but it also bounds the number of hops.  That is, the second lemma (or
rather its corollary)
states that if there is an $h$-hop path in $G$, then there is a
corresponding $\ceil{h/r}$-hop path in $H$ with the same weight.
Together, these imply that the distances computed in $H$ can be used
to compute distances in~$G$.

\begin{lemma}\lemlabel{Hpaths}
  Consider any $s_i,v_j \in V_H$.  Let $p_H$ be any $s_i$-to-$v_j$
  path in $H$.  Then there is an $s$-to-$v$ path $p$ in $G$ with $w(p)
  = w_H(p_H) - \delta_i(s) + \delta_j(v)$.

  If $w(e) \geq 0$ for all $e \in E^+$ and we consider $s,v \in V$, then
  the statement simplifies to: let $p_H$ be any $s$-to-$v$ path in
  $H$; then there is an $s$-to-$v$ path $p$ in $G$ with $w(p) =
  w_H(p_H)$. 
\end{lemma}
\begin{proof}
  The simplification follows from the main claim by observing that if
  $w(e) \geq 0$ for all $e\in E^+$, then $\delta_0(v) = 0$ for all
  $V$. Thus, for $i=0$ and $j=0$, the simplification follows.

  The proof of the main claim is by induction on $\card{p_H}$, the
  size of the path.

  The base case is an empty path from $s_i$ to $s_i$ in $H$ and the
  corresponding empty path from $s$ to $s$ in $G$.  Indeed $0 = 0 -
  \delta_i(s) + \delta_i(s)$.  

  For the inductive step, consider a nonempty path $p_H$.  Decompose
  $p_H$ into its last edge $(u_\ell,v_j)$ and the remaining subpath $p_H'$ from
  $s_i$ to $u_\ell$.  By inductive assumption, there is an $s$-to-$u$
  path $p'$ in $G$ with
  $w(p') = w_H(p_H') - \delta_i(s) + \delta_\ell(u)$.  By definition
  of $w_H$, we also have
  $w(u,v) = w_H(u_\ell,v_j) -\delta_\ell(u) + \delta_j(v)$.  We now
  have two cases depending on whether the edge is a self edge or not.

  If $u \neq v$, then $(u,v)\in E^+\cup N$ and $p$ is formed by
  appending $(u,v)$ to $p'$. In this case, we have
  \begin{align*}
    w(p) &= w(p') + w(u,v) \\
         &= (w_H(p_H') - \delta_i(s) + \delta_\ell(u)) + (w_H(u_\ell,v_j)
           -\delta_\ell(u) + \delta_j(v)) \\
         &= w_H(p_H') +  w_H(u_\ell,v_j) -\delta_i(s) + \delta_j(v) \\
         &= w_H(p_H) -\delta_i(s) + \delta_j(v) \ .
  \end{align*}

  If instead $u = v$, and the final edge is $(v_\ell,v_j)$, then the path
  $p$ is the same as the path $p'$.  Here we observe that
  $w_H(v_\ell,v_j) = 0 + \delta_\ell(v) - \delta_j(v)$, or
  $\delta_\ell(v) = w_H(v_\ell,v_j) + \delta_j(v)$.  Thus,
  \begin{align*}
    w(p) =w(p') &= w_H(p_H') - \delta_i(s) + \delta_\ell(v) \\
         &= w_H(p_H') -\delta_i(s) + w_H(v_\ell,v_j) + \delta_j(v) \\
         &= w_H(p_H) -\delta_i(s) + \delta_j(v)
  \end{align*}
\end{proof}

\begin{lemma}\lemlabel{hops-in-H}
  Let $p$ be any $h$-hop $s$-to-$v$ path in $G$, for any $s,v\in V$.
  Then there is an $h_H$-hop $s_0$-to-$v_j$ path $p_H$ in $H$, for some
  layer $0 \leq j \leq r$, with the following two properties: (1)
  $w_H(p_H) = w(p) + \delta_0(s) - \delta_j(v)$, and (2)
  $rh_H+j \leq h$.
\end{lemma}
\begin{proof}
  The proof is by induction on $\card{p}$.  The base case is an empty
  path from $s$ to itself in $G$ and the corresponding empty path
  $s_0$ to $s_0$ in $H$.

  For the inductive step, consider a path $p$, which we can decompose
  into a subpath $p'$ from $s$ to $u$ and the edge $(u,v)$.  By
  inductive assumption, there is a corresponding $h_H'$-hop path
  $p_H'$ in $H$ from $s_0$ to some $u_\ell$ with
  $w_H(p_H') = w(p') + \delta_0(s) - \delta_\ell(u)$.  There are
  several cases.

  Case 1: $(u,v) \in E^+$. Then $p'$ is an $h$-hop path, and thus the
  inductive assumption on the hops for $p_H'$ is
  $r h_H' + \ell \leq h$.  The path $p_H$ is formed by appending
  $(u_\ell,v_j)$, where $j\in\set{0,\ell}$ depends on whether $v\in R$,
  to the path $p_H'$.  We thus get
  \begin{align*}
    w_H(p_H) &= w_H(p_H') +
               w_H(u_\ell,v_j) \\
             &= (w(p') + \delta_0(s) - \delta_\ell(u)) + (w(u,v) +
               \delta_\ell(u) - \delta_j(v)) \\
             &= w(p) + \delta_0(s) -\delta_j(v) \ .
  \end{align*}
  Since $w(u_\ell,v_j) \geq 0$ by \lemref{h-weights}, the number of
  hops in $p_H$ is the same as $p_H'$.  Moreover, $j \leq \ell$.  So
  $r h_H + j \leq r h_H' + \ell \leq h$ as required.

  Case 2: $(u,v) \in E^-$.  Then $p'$ is an $(h-1)$-hop path, and thus
  the inductive assumption on the hops for $p_H'$ is $r h_H' + \ell
  \leq h-1$.\\
  Case 2a: If $\ell < r$, then the path $p_H$ is formed by appending
  $(u_\ell,v_j)$, where $j\in\set{0,\ell+1}$ depends on whether $v \in
  R$, to the path $p_H'$.  The formula for $w_H(p_H)$ is the same as
  for Case~1.  Moreover, by \lemref{h-weights}, we again have $h_H' =
  h_H$, but now $j\leq \ell+1$.  We thus have $r h_H + j \leq r h_H' +
  \ell+1 \leq (h-1)+1 = h$.\\
  Case 2b: If $\ell = r$, then the path $p_H$ is formed by appending
  two edges $(u_\ell,u_0)$ and $(u_0,v_j)$ to the path $p_H'$, where
  $j\in\set{0,1}$ depends on whether $v \in R$.   Now we have
  \begin{align*}
    w_H(p_H) &= w_H(p_H') + w_H(u_r,u_0) + w_H(u_0,v_j) \\
             &= (w(p') + \delta_0(s) - \delta_r(u)) + (0 +
               \delta_r(u) - \delta_0(u)) + (w(u,v) + \delta_0(u) -\delta_j(v)) \\
             &= w(p) + \delta_0(s) -\delta_j(v) \ .
  \end{align*}
  Here, the edge $(u_r,u_0)$ may be a negative-weight edge, but by
  \lemref{h-weights} the other edge is not. Thus, we can only conclude
  that $h_H \leq h_H' + 1$.  Nevertheless, because $j \leq 1$ and
  $\ell = r$, we have $r h_H + j \leq r(h_H' + 1) + 1= (rh_H' + \ell) +
  1 \leq (h-1) + 1 = h$, or $rh_H + j \leq h$ as claimed.
\end{proof}

\begin{corollary}\corlabel{ceil-hops}
  Let $p$ be any $h$-hop $s$-to-$v$ path in $G$, for any $s,v \in
  V$. Then for all layers $i$ with $v_i \in V_H$, there is an
  $\ceil{h/r}$-hop path $p_H$ in $H$ from $s_0$ to every $v_i$ with
  weight $w_H(p_H) = w(p) + \delta_0(s)- \delta_i(v)$.

  If $w(e) \geq 0$ for all $e \in E^+$, and we consider $i=0$, then a
  special case of the claim is: let $p$ be any $h$-hop $s$-to-$v$ path
  in $G$.  Then there is an $\ceil{h/r}$-hop path $p_h$ in $H$ with
  $w_H(p_H) = w(p)$.
\end{corollary}
\begin{proof}
  \lemref{hops-in-H} states that there exists a layer $j$ and a path
  $p_H$ from $s_0$ to $v_j$ in $H$ with (1) weight
  $w_H(p_H) = w(p) +\delta_0(s) - \delta_j(v)$, and (2) a number of
  hops $h_H$ with $rh_H + j \leq h$.

  Case 1: $j = 0$. Then we have $h_H \leq h/r \leq \ceil{h/r}$.  The
  claim can be achieved for all $i$ by appending edges
  $(v_0,v_1) , (v_1,v_2), \ldots, (v_{i-1},v_i)$ to the path $p_H$.
  By \lemref{h-weights}, the edges all have nonnegative weight, and
  hence the number of hops does not change.  Moreover, the $\delta$'s
  telescope, giving total weight
  $w(p) +\delta_0(s) - \delta_0(v) + (\delta_0(v) - \delta_1(v)) +
  (\delta_1(v) - \delta_2(v)) + \cdots + (\delta_{i-1}(v)
  -\delta_i(v)) = w(p) - \delta_0(s) - \delta_i(v)$.

  Case 2: $j \geq 1$. Then we have $rh_H + 1 \leq rh_H + j \leq h$, or
  $h_H < h/r \leq \ceil{h/r}$.  Since the inequality is strict, and
  $h_H$ is an integer, we have $h_H \leq \ceil{h/r}-1$.  To achieve
  the claim, we can therefore afford to use one more negative edge in
  $H$.  Thus, the paths to $v_i$ are formed by first appending
  $(v_j,v_{j+1}), (v_{j+1},v_{j+2}),\cdots,(v_{r-1},v_r),(v_r,v_0)$ to
  the path; by \lemref{h-weights}, only the last edge here has
  negative weight, increasing the number of hops to at most
  $\ceil{h/r}$.   As in case 1, the $\delta$'s telescope, giving a total
  weight of $w(p) + \delta_0(s) - \delta_0(v)$ to $v_0$.  To finish
  out, apply case 1 to this augmented path.

  Finally, if $w(e) \geq 0$ for all $e \in E^+$, then $\delta_0(v) =
  0$ for all $v \in V$, which gives the simplified statement.   
\end{proof}

\paragraph{Proof of \lemref{hop-reduction}.}
Let us start by considering the correctness.  Suppose that are no
negative-weight cycles in $G$; show that \algref{hopreduction} returns
a price function, and moreover that the price function is correct.
(The contrapositive says that if the algorithm reports a
negative-weight cycle, then that reporting is correct.)  If there is
no negative-weight cycle, then there exist shortest paths that are
simple paths, and hence $\dist[G]{}(V,v) = \dist[G]{\hat{k}}(V,v)$.
Let $\kappa = \ceil{\hat{k}/r}$.  Then applying \corref{ceil-hops}, we
therefore have that for all $v_i\in V$,
$\dist[H]{\kappa}(V,v_i) \leq \dist[G]{\hat{k}}(V,v) - \delta_i(v)=
\dist[G]{}(V,v) - \delta_i(v)$.  From \lemref{Hpaths}, we also have
$\dist[G]{}(V,v) \leq \dist[H]{\kappa+1}(V,v_i) + \delta_i(v) \leq
\dist[H]{\kappa}(V,v_i) +\delta_i(v)$.  Thus, the distances must be
the same. That is,
$\dist[H]{\kappa+1}(V,v_i) = \dist[H]{\kappa}(V,v_i) =
\dist[G]{}(V,v)-\delta_i(v)$.  Therefore, (1) the algorithm does not
report a cycle, and (2) for all $v\in V$,
$\phi(v) = d(v) = \dist[H]{\kappa}(V,v) = \dist[G]{}(V,v)$ is the same
price function from Johnson's strategy~\cite{Johnson77}, and hence
$\phi$ is a valid price function that eliminates all negative
edges~$N$.

Next consider the case that $G$ does contain a negative-weight cycle.
Then by \lemref{hops-in-H}, there is a negative-weight cycle in~$H$,
and moreover there is such a cycle that includes some layer-0
vertex~$v_0$.  Observe that if
$\dist[H]{\kappa+1}(V,u_j) = \dist[H]{\kappa}(V,u_j)$ for all
$u_j \in V_H$, then it must be the case that that these are the actual
shortest-path distances, i.e.,
$\dist[H]{\kappa}(V,u_j) = \dist[H]{}(V,u_j)$.  Given the presence of
a negative-weight cycle, however, we know that
$\dist[H]{\kappa}(V,v_0) \neq \dist[H]{}(V,v_0)$.  Thus, there must
exist some $u_j \in V_H$ with $\dist[H]{\kappa+1}(V,u_j) <
\dist[H]{\kappa}(V,u_j)$, and \algref{hopreduction} reports a cycle.

Now let us consider the running time, which is dominated by two
super-source shortest path computations.  The first computation is the
$\leq r$-hop distances $\delta_i$ in $G$.  Because $w(e)\geq 0$, \lemref{BFD}
states that these can all be computed in a total of $O(r m\log n)$
time. 

The shortest-path computation in $H$ has a running time that depends
on the size of $H$.  Let $X$ be the set of negative vertices, i.e.,
$N = \id{out}(X)$.  Moreover, because the only negative-weight edges
are in $N$, it follows that $R = \reach[G]{r}(V) = \reach[G]{r}(N)$.
Thus, by assumption that $X$ is $r$-remote, we have
$\card{R} \leq n/r$, where $n=\card{V}$.  Now consider the
construction of~$H$.  We directly get
$\card{V_H} = r \cdot \card{R} + n \leq r\cdot n/r + n = 2n$. As for
the edges, by construction each vertex $u_j\in V_H$ has at most one
outgoing edge $(u_j,v_j)$ corresponding to the edge
$(u,v)\in E^+ \cup N$, plus one self edge.  Applying the simplifying
assumptions that all vertices have degree at most $O(m/n)$ and
$m\geq 2n$,\footnote{Specifically, that the number of nonnegative
  outgoing edges is $O(m/n)$} we have $\card{E_H} \leq \card{V_H} \cdot O(m/n) = O(m)$.

We conclude by applying \lemref{BFD} for the cost of computing
$\ceil{\hat{k}/r}$-hop distances in $H$.  Because $H$ has $O(n)$ vertices and
$O(m)$ edges, the running time of this step is $O(\ceil{\hat{k}/r}
m\log n)$.  Adding the running time of the shortest paths in $G$, we
get $O((\hat{k}/r)m\log n + rm\log n)$; the ceiling can be dropped
because the second term subsumes the first when $\hat{k}$ is small. 
\qed